\documentclass[%
 reprint,
 amsmath,amssymb,
 aps,
pra,
]{revtex4-2}

\usepackage{graphicx}% Include figure files
\usepackage{dcolumn}% Align table columns on decimal point
\usepackage{bm}% bold math
\usepackage[dvipsnames]{xcolor}
\usepackage{physics}
\usepackage{orcidlink}
\usepackage{mathrsfs}
\usepackage{hyperref}
\hypersetup{
    colorlinks = true,
    citecolor = magenta,
    linkcolor = magenta,
}% you can set the hyperlink colors here

\usepackage{algorithm}
\usepackage{algpseudocode}
\usepackage{amsthm}
\usepackage{braket}
\usepackage[english]{babel}
\usepackage{enumitem}
\newtheorem{theorem}{Theorem}[section]

\newtheorem{proposition}{Proposition}[section]
\newtheorem{lemma}{Lemma}[section]
\newtheorem{corollary}{Corollary}[section]
\newtheorem{definition}{Definition}[section]

\newtheorem{maintheorem}{Theorem}
\newtheorem{maindefinition}{Definition}
\newtheorem{mainproposition}{Proposition}
\newif\ifenablecommands  % Define a conditional variable
\enablecommandsfalse

\begin{document}

\preprint{APS/123-QED}

%\title{Bounds on long-range conditional mutual information}
%\title{Information-theoretic phases of decohered Gibbs states and hidden Markov networks}

% \title{Long-Range Conditional Mutual Information and Information-theoretic phases of Quantum hidden Markov networks}

\title{Conditional Mutual Information and Information-Theoretic Phases of Decohered Gibbs States}

%\title{Dissipation cannot induce long-range conditional correlations at high temperature}

%\title{Long-Range Conditional Mutual Information in Classical and Quantum Hidden Markov Networks}

\author{Yifan Zhang}
\email{yz4281@princeton.edu}
\affiliation{Department of Electrical and Computer Engineering, Princeton University, Princeton, NJ 08544}
 
\author{Sarang Gopalakrishnan}
\email{sgopalakrishnan@princeton.edu}
\affiliation{Department of Electrical and Computer Engineering, Princeton University, Princeton, NJ 08544}
\date{\today}% It is always \today, today,

\begin{abstract}

Classical and quantum Markov networks---including Gibbs states of commuting local Hamiltonians---are characterized by the vanishing of conditional mutual information (CMI) between spatially separated subsystems. Adding local dissipation to a Markov network turns it into a \emph{hidden Markov network}, in which CMI is not guaranteed to vanish even at long distances. The onset of long-range CMI corresponds to an information-theoretic mixed-state phase transition, with far-ranging implications for teleportation, decoding, and state compressions. Little is known, however, about the conditions under which dissipation can generate long-range CMI. In this work we provide the first rigorous results in this direction. We establish that CMI in high-temperature Gibbs states subject to local dissipation decays exponentially, (i) for classical Hamiltonians subject to arbitrary local transition matrices, and (ii) for commuting local Hamiltonians subject to unital channels that obey certain mild restrictions. Conversely, we show that low-temperature hidden Markov networks can sustain long-range CMI. Our results establish the existence of finite-temperature information-theoretic phase transitions even in models that have no finite-temperature thermodynamic phase transitions. We also show several applications in quantum information and many-body physics.

\end{abstract}

\maketitle

\ifenablecommands
\section{Introduction (bookmark purpose)}
\fi

It is generally believed that long-range correlation quantities cannot be generated under finite-time dynamics because of the Lieb-Robinson bound. However, such a belief does not hold true in \emph{conditional mutual information} (CMI), a multipartite correlation measure. Unlike standard measures of entanglement and mutual information, the CMI does \emph{not} obey a light cone \cite{lee2024universal}, in quantum or classical systems, and can be generated using local channels. In the quantum context, the ability of local channels to generate long-range CMI can be seen as a form of generalized teleportation: informally, if Bob separately shares a Bell pair with Alice and one with Claire, he can implement a Bell measurement and induce entanglement between Alice and Claire that is conditional on the measurement outcome \cite{zhang2024nonlocal}.

The generation of long-range CMI under local or finite-time dynamics is closely related to quantum error corrections and mixed-state phase transitions. Under local noisy dynamics, error accumulates in the code and destroys the codeword after an $O(1)$ time threshold. This threshold transition challenges the standard paradigms of statistical mechanics as correlation length cannot diverge in $O(1)$ time. However, noise accumulation can generate long-range CMI, which operationally reflects the failure of quasi-local decoders after passing the threshold. Similar argument also shows that long-range CMI is a crucial ingredient in all mixed-state phase transitions \cite{sang2024stability}. However, very little is systematically known about the states that can generate long-range CMI under local decoherence.

%
% If a quantum error correcting code is subject to noise, errors accumulate; at some finite time the density of errors passes the threshold at which logical information becomes irretrievable. This threshold transition challenges the standard paradigms of statistical mechanics: one expects phase transitions to be associated with diverging timescales rather than occurring at a finite $O(1)$ time. Indeed, the Lieb-Robinson bound (which also holds for quantum channels) seems to prevent the divergence of conventional correlation lengths at a finite timescale. It was recently appreciated that the most natural intrinsic quantity of an error correcting code that sees a divergent length-scale at the error correction threshold transition is the quantum \emph{conditional mutual information} (CMI), a multipartite entanglement measure. Unlike standard measures of entanglement and mutual information, the CMI does \emph{not} obey a light cone**, in quantum or classical systems. In the quantum context, the ability of local decoherence to generate nonlocal CMI can be seen as a form of generalized teleportation: informally, if Bob separately shares a Bell pair with Alice and one with Claire, he can implement a Bell measurement and induce entanglement between Alice and Claire that is conditional on the measurement outcome. The argument in Ref. ** suggests that long-range CMI is a crucial ingredient in all mixed-state phase transitions. However, very little is systematically known about the states that can generate long-range CMI. 

The bulk of this work is dedicated to establishing that local decoherence \emph{cannot} induce long-range CMI in a wide class of high-temperature Gibbs states. Our considerations apply to both quantum and classical CMI. The classical states we consider are \emph{Markov networks}, states with zero CMI, but local dissipation generates CMI and turns them into \emph{hidden Markov networks}.
%
% , which dissipation turns into hidden Markov networks; 
%
we call the analogous quantum states \emph{quantum hidden Markov networks}. In both cases, we construct low-temperature states that provably generate long-range CMI under decoherence. Together, our results imply the existence of nonzero-temperature information-theoretic transitions in the structure of Gibbs states, even when there is no thermodynamic phase transition at nonzero temperature: dissipation generates long-range CMI only when the initial state is below a critical temperature.

\begin{figure*}[t]
\includegraphics[width=\linewidth]{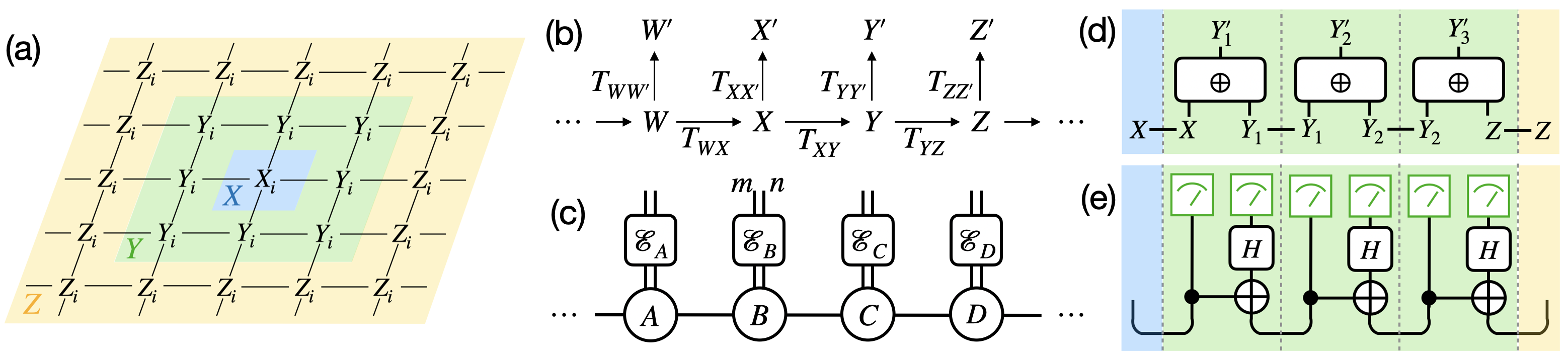}
\caption{\label{fig:setup}(a) An example of partitioning a Markov network exhibiting conditional independence. $\textcolor{Green}{Y}$ separates $\textcolor{Cerulean}{X}$ and $\textcolor{Goldenrod}{Z}$ and we have $I(\textcolor{Cerulean}{X}:\textcolor{Goldenrod}{Z}|\textcolor{Green}{Y})=0$. (b) An example of hidden Markov chain. $WXYZ$ forms a Markov chain and each site is subject to a transition matrix. (c) An example of quantum hidden Markov chain. $\rho_{ABCD}$ forms a quantum Markov chain and each site is subject to a quantum channel. Note that we use the double leg to denote the bra and ket, so fixing the index $mn$ corresponds to choosing $\ket{m}\bra{n}$. (d) An example of classical hidden Markov network. Grey dashed lines separate sites and we partition the system into $\textcolor{Cerulean}{X}$, $\textcolor{Green}{Y}$, and $\textcolor{Goldenrod}{Z}$. The white boxes represent a channel that takes two bits as input and outputs their parity. (e) An example of quantum hidden Markov network. Grey dashed lines separate sites and we partition the system into $\textcolor{Cerulean}{A}$, $\textcolor{Green}{B}$, and $\textcolor{Goldenrod}{C}$. The circuit rotates the computation basis to the Bell basis and then measures it.}
\end{figure*}

\emph{Definitions}.---We start with the classical case.
%
%\emph{Conditional mutual information} (CMI),  
The CMI $I(\textcolor{Cerulean}{X}:\textcolor{Goldenrod}{Z}|\textcolor{Green}{Y})$ is the mutual information between the two random variables 
%
%measures the correlation between two random variables, 
$\textcolor{Cerulean}{X}$ and $\textcolor{Goldenrod}{Z}$, conditioned on a third variable, $\textcolor{Green}{Y}$. In terms of Shannon entropies, $I(\textcolor{Cerulean}{X}:\textcolor{Goldenrod}{Z}|\textcolor{Green}{Y}) = H(\textcolor{Cerulean}{X}\textcolor{Green}{Y}) + H(\textcolor{Goldenrod}{Z}\textcolor{Green}{Y}) - H(\textcolor{Green}{Y}) - H(\textcolor{Cerulean}{X}\textcolor{Green}{Y}\textcolor{Goldenrod}{Z})$. Note that this definition also carries over to the quantum setting if one replaces all Shannon entropies with Von Neumann entropies. If $I(\textcolor{Cerulean}{X}:\textcolor{Goldenrod}{Z}|\textcolor{Green}{Y})=0$, a property known as conditional independence, the variables $\textcolor{Cerulean}{X}$, $\textcolor{Green}{Y}$, and $\textcolor{Goldenrod}{Z}$ are said to form a (classical or quantum) \emph{Markov chain}. Intuitively, this means that all correlations between $\textcolor{Cerulean}{X}$ and $\textcolor{Goldenrod}{Z}$ are mediated by $\textcolor{Green}{Y}$. %, making $\textcolor{Green}{Y}$ a kind of information bridge that decouples the two other variables. 
In higher dimensions, Markov chains generalize to \emph{Markov networks}, where a lattice of random variables exhibits zero CMI when partitioned such that $\textcolor{Green}{Y}$ separates $\textcolor{Cerulean}{X}$ and $\textcolor{Goldenrod}{Z}$ (see Fig. \ref{fig:setup}(a)). \emph{Quantum Markov chain} and \emph{Markov networks} are analogously defined\cite{hayden2004structure,leifer2008quantum}.
%
%CMI is a multipartite correction measure. 
Because Markov networks represent distributions without multipartite correlations, they are highly structured and computationally tractable. They have been extensively studied across physics, probability theory, and machine learning. A key property of classical Markov chains is the Hammersley-Clifford theorem, which establishes the equivalence between Markov networks and Gibbs states of local Hamiltonians. However, while quantum Markov networks and Gibbs states of commuting Hamiltonians on Hypercube lattices are equivalent \cite{brown2012quantum}\footnote{There exist quantum Markov networks on graphs with triangles that cannot be written as Gibbs states of commuting local Hamiltonians \cite{brown2012quantum}. However, this non-commuting property disappears after ``coarse-graining'' by combining multiple adjacent sites into a super-site. To our knowledge, no examples exist of quantum Markov networks that are not Gibbs states of commuting local Hamiltonians even after coarse-graining.}, generic Gibbs states of local quantum Hamiltonians do not obey the Hammersley-Clifford theorem.  %,  highlighting the broad relevance of Markov networks across different fields.
%

%The concept of Markov chains and networks extends naturally to the quantum realm. Here, quantum CMI $I(\textcolor{Cerulean}{A}:\textcolor{Goldenrod}{C}|\textcolor{Green}{B})$ on a tripartition $\textcolor{Cerulean}{A}\textcolor{Green}{B}\textcolor{Goldenrod}{C}$ and quantum conditional independence are well-defined \cite{hayden2004structure,leifer2008quantum}, leading to \emph{quantum Markov chains} and \emph{quantum Markov networks}. These structures have been studied for decades in relation to quantum information tasks such as constructing recovery maps \cite{fawzi2015quantum,wilde2015recoverability,sutter2016universal,junge2018universal} and reconstructing many-body states from their local marginals~ \cite{shi2021domain,shi2021entanglement,kim2022modular,kim2023universal,kim2024chiral,kim2024classifying,kim2024conformal,kim2021entropy,kim2021entropy2}. More recently, as mentioned above, CMI has been explored as a diagnostic of mixed-state phase transitions. 

Local dissipation does not preserve the Markov chain property: instead, it turns Markov chains (networks) into \emph{hidden Markov chains (networks)} (Fig. \ref{fig:setup}(b)), for which the CMI is generically positive. In classical probability, hidden Markov networks are powerful tools for modeling complex distributions due to their ability to encode non-local, multipartite correlations.
Quantum Hidden quantum Markov chains/networks are defined analogously (Fig. \ref{fig:setup}(c))
%
%Quantum Markov networks similarly transform into \emph{quantum hidden Markov networks} under local channels 
~\footnote{We note that the notion of Hidden Quantum Markov Model is defined in the literature \cite{monras2010hidden}. Their quantum generalization of the hidden Markov model is also well-motivated but is different from ours. To our best knowledge, the quantum generalization of hidden Markov networks has not been discussed in literature}, and again have nonvanishing CMI. The existence of long-range CMI in the quantum setting is associated with phenomena that have no classical counterpart, such as quantum teleportation and measurement-induced entanglement \cite{verstraete2004entanglement,popp2005localizable,bao2024finite,napp2022efficient,lin2023probing,google2023measurement,cheng2024universal,zhang2024nonlocal,mcginley2024measurement}. Which quantum states can potentially serve as resources for these protocols,  however, remains an open question. 

The final concept we need to define is the \emph{Markov length} $\xi$. In general, a hidden Markov model might not obey an \emph{exact} Markov condition; however, we say that this condition is obeyed \emph{approximately} if $I(X:Z|Y) \propto e^{-|Y|/\xi}$ where $|Y|$ is the shortest path length across $Y$ and $\xi$ is the Markov length. If there is a uniformly defined Markov length (independent of the size of $Y$) then the Markov chain property is satisfied to any desired accuracy after a small amount of coarse-graining. When we speak of ``long-range'' CMI we mean that the Markov length diverges.

\ifenablecommands
\section{Examples of Quantum Hidden Markov Networks (bookmark purpose)}
\fi
\emph{Examples of Quantum Hidden Markov Networks}.---We provide examples of classical and quantum hidden Markov models in Fig. \ref{fig:setup}(d,e), where we also show the relation between quantum hidden Markov networks and teleportation. In Fig. \ref{fig:setup}(d), neighboring sites share two identical bits, with a 50\% probability of being either zero or one. These correlated bits are represented by\(X\), \(Y_1\), \(Y_2\), \(Y_3\), and \(Z\). The system forms a Markov chain initially. To generate a hidden Markov model, a channel is applied at each site in \(B\)that only retains the parity of the two bits (e.g., \(Y_1' = X \oplus Y_1\)). The resulting distribution of \(X\), \(Y_1'\), \(Y_2'\), \(Y_3'\), and \(Z\) satisfies the following parity constraint but is otherwise constrained.
\begin{equation}
    X \oplus Y_1' \oplus Y_2' \oplus Y_3' \oplus Z = 1.
\end{equation}
Therefore, by knowing the values of \(Y_1'\), \(Y_2'\), and \(Y_3'\), the parity of $X \oplus Z$ is fixed, so \(X\) and \(Z\) become maximally correlated. This results in \(I(X:Z|Y_1'Y_2'Y_3') = 1\) across long distances.

The quantum analogue of this model is shown in Fig. \ref{fig:setup}(e), which coincides with the quantum network circuit. The neighboring sites share a Bell pair, and the channel corresponds to a quantum instrument \(\mathcal{E}\) that performs measurements in the Bell basis. Equivalently, the quantum instrument measures the \(ZZ\) and \(XX\) stabilizers, while in the classical example only $ZZ$ stabilizers are measured. Knowing the measurement outcomes in \(B\), \(A\) and \(C\) share a maximally entangled state, resulting in \(I_{\mathcal{E}[\rho]}(A:C|B) = 2\) across long distances. This example demonstrates the relation between teleportation and CMI. 

In both cases, the initial Markov chain can be viewed as the ground state of a nearest-neighbor Hamiltonian: the \(ZZ\) stabilizer in the classical case and both the \(ZZ\) and \(XX\) stabilizers in the quantum case. Increasing the temperature introduces noise, degrading the correlations in the classical bits and the Bell states. As a result, the CMI of the hidden Markov models decays exponentially with increasing distance between \(A\) and \(C\). Thus truly long-range CMI is absent at any finite temperature in this example. A natural question is whether long-range CMI is absent at nonzero temperatures \emph{in general}. In what follows we establish that long-range CMI is indeed absent at sufficiently high temperatures, but also that there are models in which it is stable for a range of low but nonzero temperatures.

%In other words, the long-range CMI here is unstable against increasing temperatures.

%We aim to investigate whether hidden Markov networks can exhibit long-range CMI that remains stable against temperature perturbations. In the next section, we rule out this possibility at sufficiently high temperatures.
% \begin{figure}
% \includegraphics[width=0.8\linewidth]{hmm_ex.png}
% \caption{\label{fig:hmm_ex}(a) An example of classical hidden Markov network. Grey dashed lines separate sites and we partition the system into $\textcolor{Cerulean}{X}$, $\textcolor{Green}{Y}$, and $\textcolor{Goldenrod}{Z}$. The white boxes represent a channel that takes two bits as input and outputs their parity. (b) An example of quantum hidden Markov network. Grey dashed lines separate sites and we partition the system into $\textcolor{Cerulean}{A}$, $\textcolor{Green}{B}$, and $\textcolor{Goldenrod}{C}$. The circuit rotates the computation basis to the Bell basis and then measures it.}
% \end{figure}

\ifenablecommands
\section{Finite Markov Length at High Temperature (bookmark purpose)}
\fi
\emph{Finite Markov Length at High Temperature}.---We present our first result that addresses the lack of long-range CMI at high temperatures. Let $H=\sum_a \lambda_a h_a$, where $h_a$ are Hermitian matrices with local support and $\lambda_a$ are real coefficients. We require that $h_a$ are products of some local operator basis such as Pauli operators. $H$ is a commuting Hamiltonian if any $h_a$ and $h_b$ commute. Let $\rho \propto e^{-\beta H}$ be the Gibbs state at temperature $\beta$. 

We will consider the \emph{commutation-preserving} channels defined below. 
\begin{maindefinition}
\label{def:comm_preserv} Given a Hamiltonian $H=\sum_a \lambda_a h_a$ where every pair of $h_a$ and $h_b$ commutes. Consider the set of operators $\{O_m\}$, where each $O_m$ is a product of $h_a$, namely $O_m=\prod_a h_a^{\mu_{a,m}}$, where $\mu_{a,m}$ is a non-negative integer denoting the multiplicity of  $h_a$. Consider a set of channels $\{\mathcal{E}_i\}$ where each $\mathcal{E}_i$ acts on site $i$. The set is commutation-preserving if for any operators $O_m$ and $O_n$ and for any subset of sites $S$, $(\otimes_{i \in S} \mathcal{E}_i)[O_m]$ and $(\otimes_{i \in S} \mathcal{E}_i)[O_n]$ commute.
\end{maindefinition}
We also impose that the local channels are \emph{unital}, namely $\mathcal{E}_i[I]=I$ for all $i$. This category of channels contains many physically-relevant ones. For example, if $h_a$ are Pauli operators, then bit-flip, dephasing, and depolarization channels at any rate and their composition all belong to this category.

We present our first main result below.
\begin{maintheorem}\label{thm:decay_cmi_informal}
    (Informal) Consider the Gibbs state $\rho \propto e^{- \beta H}$ of some local commuting Hamiltonian $H$. Let $\mathcal{E} = \otimes_i \mathcal{E}_i$, where $ \mathcal{E}_i$ is a unital channel acting on site $i$ We demand that $\{\mathcal{E}_i\}$ is commutation-preserving. Consider any three subsystems $ABC$. When $\beta$ is less than a critical temperature $\beta_c = O(1)$, we have the following decay of CMI:
    \begin{equation}
        I_{\mathcal{E}[\rho]}(A:C|B) \propto e^{-\frac{d_{AC}}{\xi}}
    \end{equation}
    Where $d_{AC}$ is the distance between $A$ and $C$, and $\xi=O(\beta/\beta_c)$ is called the Markov length
\end{maintheorem}
Our results implies that quantum hidden Markov networks under unital, commutation preserving channels have a $O(1)$ Markov length, and thus cannot support long-range CMI at high temperautre.

Our major technical contribution is to show that the generalized cluster expansion, a standard tool to analyze high-temperature Gibbs states, can be applied to $\mathcal{E}[\rho]$ as well. We defer the proof to Appendix \ref{high_temp_appn} and explain the rough idea here. Without loss of generality, we restrict $\mathcal{E}$ to act on $B$ only since applying local channels to $A$ and $C$ only decreases CMI due to the data processing inequality. We will consider the following matrix $H(A:C|\mathcal{E}[B])$:
\begin{equation}\label{eq:mcmi_main_text}
\begin{split}
    H(A:C|\mathcal{E}[B]) =  &\log(\mathcal{E}[\rho_{AB}]) + \log(\mathcal{E}[\rho_{BC}])\\
    - &\log(\mathcal{E}[\rho_{B}])  - \log(\mathcal{E}[\rho_{ABC}])
    \end{split}
\end{equation}
Where $\rho_{L}$ denotes the reduced density matrix in the region $L$. One can quickly see that ${ I_{\mathcal{E}[\rho]}(A:C|B) = -\text{Tr}\,[\mathcal{E}[\rho]H(A:C|\mathcal{E}[B])]}$, thus the operator norm $\norm{H(A:C|\mathcal{E}[B])}$ upper-bounds the CMI. 

To upper-bound $\norm{H(A:C|\mathcal{E}[B])}$, we Taylor expand $\norm{H(A:C|\mathcal{E}[B])}$ in $\beta$
\begin{equation}\label{eq:mcmi_taylor_expand}
    H(A:C|\mathcal{E}[B]) = \sum_{m=0}^{\infty} \frac{(-\beta)^m}{m!} \frac{\partial^m}{\partial \beta^m} H(A:C|\mathcal{E}[B])
\end{equation}
Cluster expansion reorganizes the above derivative in $\beta$ into derivatives of Hamiltonian coefficients $\lambda_a$, namely
\begin{equation}
    \beta^m \frac{\partial^m}{\partial \beta^m} H(A:C|\mathcal{E}[B]) = \left(\sum_a \lambda_a \frac{\partial}{\partial \lambda_a}\right)^m H(A:C|\mathcal{E}[B])
\end{equation}
We will use the cluster expansion to show that 
\begin{enumerate}
    \item $\frac{1}{m!}\norm{\frac{\partial^m}{\partial \beta^m} H(A:C|\mathcal{E}[B])}$ decays as $c^m$, where $c$ is an $O(1)$ constant, so Eq. (\ref{eq:mcmi_taylor_expand}) converges at high temperature where $\beta < c$.
    \item The first nonzero $\frac{1}{m!}\norm{\frac{\partial^m}{\partial \beta^m} H(A:C|\mathcal{E}[B])}$ is when $m \ge d_{AC}$. This establishes the main result of finite Markov length.
\end{enumerate}
Our proof strategy only works for the restricted class of Gibbs states and noise channels defined above. The difficulty in generalizing our results to non-unital channels and non-commuting Hamiltonians seem essentially technical, and we have not found any candidate physical mechanism that would generate long-range CMI in generic high-temperature Gibbs states for generic one-site channels. We conjecture that our conclusions hold in these cases also. 
%
%, the difficulty in generalizing the proof to non-unital channels seems to originate from the proof technique, not physical  barriers. Therefore, we conjecture that the result still holds even in the case of non-unital channels and leave its proof to future work. In addition, our technique fails when $\mathcal{E}[\rho]$ possess non-commuting structures. This is why we restrict the channels to be commutation-preserving. 
(We note that the failure of our proof in non-commuting settings is related to the flaw in \cite{kuwahara2020clustering}.)

In classical systems, the notion of channels is replaced by transition matrices. There we overcome the unitality restriction (commutation is also automatic) and show the decay of CMI under any transition matrices.
\begin{maintheorem}\label{thm:classical decay_cmi_informal}
    (Informal) Consider the Gibbs distribution $\rho \propto e^{- \beta H}$ of some local Hamiltonian $H$ that is diagonal in the computational basis. Let $\mathcal{T} = \otimes_i \mathcal{T}_i$, where each $ \mathcal{T}_i$ is a transition matrix acting on site $i$. Consider any three subsystems $XYZ$. When $\beta$ is less than a critical temperature $\beta_c = O(1)$, we have the following decay of CMI:
    \begin{equation}
        I_{\mathcal{T}[\rho]}(X:Z|Y) \propto e^{-\frac{d_{XZ}}{\xi}}
    \end{equation}
    Where $d_{XZ}$ is the distance between $X$ and $Z$, and $\xi=O(\beta/\beta_c)$ gives the Markov length
\end{maintheorem}
Therefore, the existence of long-range CMI in high-temperature classical hidden Markov network is ruled out. To prove the above theorem, we exploit the fact that classical CMI is an average of post-selected mutual information, which helps us overcome the unitality restriction. We present the proof in Appendix \ref{classical_high_temp_appn}.

\ifenablecommands
\section{Long-range CMI at Low Temperature (bookmark purpose)}
\fi
\emph{Long-range CMI at Low Temperature}.---  We have established the finite Markov length at high temperature, On the contrary, we provide a family of quantum hidden Markov networks that support long-range CMI below a temperature threshold. This family consists of cluster states implementing fault-tolerant MBQCs. A cluster state, defined on a simple graph, can be considered as the ground state of the following stabilizer Hamiltonian:
\begin{equation}\label{eq:cluster_hamiltonian}
    H = \sum_{i} \left( X_i \prod_{j \setminus i} Z_j \right)
\end{equation}
Where $i$ labels the vertex on the graph and $j \setminus i$ denotes all the vertices adjacent to $i$. Cluster states are short-range entangled states that can support universal MBQC \cite{raussendorf2001one}. By choosing a cluster state with an appropriate underlying graph and measuring each qubit in some appropriate local basis, one can implement any quantum circuit as a MBQC protocol. 

In MBQC, measuring the bulk of a cluster state generates entanglement between boundaries. This fits nicely into quantum hidden Markov networks. The initial cluster state is the unique ground state of a commuting Hamiltonian, thus is a quantum Markov network. The process of measuring each qubit in the appropriate basis maps to applying local dephasing channels in the corresponding basis. By conditioning on the measurement outcome, the entanglement is reflected in the CMI.

Crucially, it is known that MBQC can implement quantum error correcting codes, preserving long-range entanglement even when the underlying cluster states are subject to local noise channels \cite{raussendorf2005long,raussendorf2007topological,raussendorf2007fault}. In particular, the ``foliation'' technique allows one to encode the syndrome measurement circuit of any Calderbank, Shor, and Steane (CSS) code into a cluster state \cite{bolt2016foliated}. This robustness allows us to show the existence of long-range CMI at low temperature.

We will consider the following setup, shown in Fig. \ref{fig:mbqc}(a). We consider a cluster state that foliates a CSS code (see the construction in Ref. \cite{bolt2016foliated}). We denote the two boundaries as \(\textcolor{Cerulean}{A}\) and \(\textcolor{Goldenrod}{C}\) and denote the bulk as \(\textcolor{Green}{B}\). As an example, Fig. \ref{fig:mbqc}(b) shows the unit cell of the cluster state the encodes the Toric code. Measurement of \(\textcolor{Green}{B}\) in the $X$ basis generates entangled code states between \(\textcolor{Cerulean}{A}\) and \(\textcolor{Goldenrod}{C}\). In the corresponding quantum hidden Makov network, we apply the channel $\mathcal B = \otimes_{i \in  \textcolor{Green}{B}} \mathcal{B}_{i}$, where each $\mathcal{B}_{i}$ is a complete bit-flip channel acting on qubit $i$.
\begin{equation}
    \mathcal{B}_{i}[\rho] = \frac{1}{2} \rho + \frac{1}{2} X_i \rho X_i 
\end{equation}
We note that $\mathcal B$ is unital and commutation-preserving with respect to stabilizer Hamiltonians, so it satisfies the condition in the previous section. Suppose the underlying quantum error correcting code has $k$ logical qubits, then measuring \(\textcolor{Green}{B}\) generates entangled logical states between \(\textcolor{Cerulean}{A}\) and \(\textcolor{Goldenrod}{C}\). Therefore, this quantum hidden Markov network has a long-range CMI lower-bounded by $2k$ because the logical subspace is maximally entangled.
\begin{equation}
    I_{\mathcal{B}[\rho]}(\textcolor{Cerulean}{A}:\textcolor{Goldenrod}{C}|\textcolor{Green}{B}) \ge 2k
\end{equation}
\begin{figure}
\includegraphics[width=\linewidth]{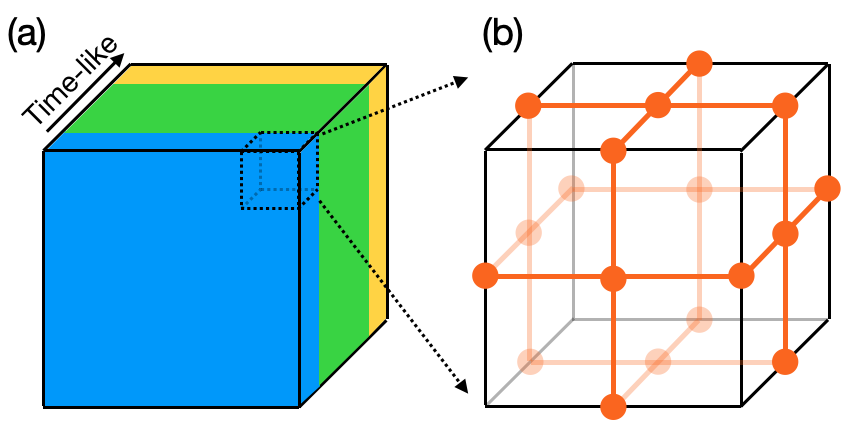}
\caption{\label{fig:mbqc}(a) Geometry of the MBQC cluster state we consider. The bulk region \(\textcolor{Green}{B}\)  is measured, generating entangled code states between \(\textcolor{Cerulean}{A}\) and \(\textcolor{Goldenrod}{C}\) which sit at the boundary. We show a three-dimensional cluster state as an example, but in general the dimensionality is arbitrary as long as a ``time-like'' direction exist, along which quantum information propagates. (b) The unit cell of the three-dimensional cluster state that encodes a Toric code in MBQC. The \textcolor{orange}{graph} that defines the cluster state has vertices on the edge and the face of the unit cell.}
\end{figure}

We now consider the above quantum hidden Markov network at finite temperatures and ask whether the above long-range CMI is stable. An important observation, first pointed out in Ref. \cite{raussendorf2005long}, is that the Gibbs state of the Hamiltonian in Eq. (\ref{eq:cluster_hamiltonian}) can also be generated by applying local dephasing channels to the clean cluster state.
\begin{mainproposition}
\label{prop:cluster_gibbs_noise}
(Eq. (5) in \cite{raussendorf2005long}) Given a cluster state $\ket{\psi}_{ABC}$, consider its finite-temperature Gibbs state $\rho^{(\beta)}_{ABC}$ with parent Hamiltonian Eq. (\ref{eq:cluster_hamiltonian}). Then $\rho^{(\beta)}_{ABC}$ can be generated from $\ket{\psi}_{ABC}$ using local dephasing channels.
    \begin{align}
        \rho^{(\beta)}_{ABC} = (\mathcal{N}_1 \otimes \mathcal{N}_2 \otimes \ldots \otimes \mathcal{N}_n) [\ket{\psi}\bra{\psi}_{ABC}]
    \end{align}
    Where $\mathcal{N}_i[\rho] = (1-p) \rho + p Z_i \rho Z_i$ denotes the dephasing channel on site $i$ with a rate $p = 1/(e^{2\beta}+1)$.
\end{mainproposition}
Therefore, a finite-temperature cluster state is equivalent to a cluster state subject to local noise. This allows us to invoke the existence of fault-tolerant MBQC to establish long-range CMI in the corresponding low-temperature quantum hidden Markov model. 
\begin{maintheorem}\label{thm:long_range_cmi_informal}
(Informal) Given a cluster state $\ket{\psi}_{ABC}$ encoding some fault-tolerant MBQC circuit with $k$ logical qubits. $B$ denotes the region to be measured in some local basis. Under the noisy measurement channel with measurement error rate $p$, suppose that $\ket{\psi}_{ABC}$ has a logical error rate $q$ that scales as $\exp(-\omega(n))$
    
Consider the corresponding finite-temperature cluster state $\rho^{(\beta)}_{ABC}$ with $\beta = \frac{1}{2} \log(p^{-1}-1)$ and apply product of local dephasing channel acting on $B$ under the corresponding measurement basis $\mathcal{E} = \mathcal{E} _1 \otimes \mathcal{E} _2 \otimes \ldots \otimes \mathcal{E} _{|B|}$. After applying the channel, the quantum hidden Markov network $\mathcal{E}_B[\rho^{(\beta)}_{ABC}]$ satisfies the following lower bound on its long-range CMI:
    \begin{align}
        I_{\mathcal{E}_B[\rho^{(\beta)}_{ABC}]}(A:C|B) \ge 2k - \epsilon
    \end{align}
Where $\epsilon$ scales as $\exp(-\omega(n))$
\end{maintheorem}
The above theorem essentially generalizes the argument in \cite{raussendorf2005long} to arbitrary fault-tolerant MBQCs, and we defer the proof to Appendix \ref{low_temp_appn}. Roughly speaking, the logical error rate is the average fidelity of the post-measurement, post-error-correction code state in $AC$ to the maximally-entangled code state, averaged over all possible measurement outcomes and all noise patterns. We rigorously define the logical error rate $q$ in Appendix \ref{low_temp_appn} and note that this definition is lower-bounded by the logical error rate in the standard Clifford simulations \cite{gidney2021stim}. There exists extensive numerical and analytic studies which confirm the exponentially small logical error rate in fault-tolerant MBQC \cite{raussendorf2007fault,raussendorf2007topological,claes2023tailored,lee2022universal,van2024fault}.

To summarize the proof, we apply the Fannes–Audenaert inequality which entails that two states that are exponentially close in trace distance are also exponentially close in their entropies. Applying it to the code space and realizing that two sets of entangled $k$ logical qubits have a mutual information of $2k$, we obtain the result.

Lastly, we comment on the recent Ref. \cite{negari2024spacetime} which proposes a similar idea of using Markov length to detect phase transitions in cluster states encoding logical information. we show that our results are complementary to theirs in Appendix \ref{comp_negari}.

\ifenablecommands
\section{Discussion (bookmarking purpose)}
\fi
\emph{Discussion}.---Our first main result is establishing that CMI decays exponentially with a finite Markov length in decohered high-temperature Gibbs states, (i)~for general commuting Hamiltonians subject to single-site unital commutation-preserving channels, and (ii)~for classical Hamiltonians subject to arbitrary single-site decoherence. Therefore, by the results of Ref.~\cite{sang2024stability}, all of these states lie in the same trivial information-theoretic phase as their parent high-temperature Gibbs states. Our second main result shows that this is not the only phase: low-temperature Gibbs states of certain commuting Hamiltonians are capable of fault-tolerant MBQC and consequently their Markov length can diverge under decoherence. Our results establish that finite-temperature information-theoretic transitions are possible even in the absence of any thermodynamic phase transition. 

Our results are relevant to the recent interest in identifying space-time fault-tolerance as a new phase of matter \cite{negari2024spacetime}. 
%
%We have systematically analyzed the long-range CMI in classical and quantum hidden Markov networks. We show that CMI decays exponentially with a finite Markov length in high-temperature hidden Markov network subject to certain physically-relevant unital channels. We strengthen this result to all transition matrices in classical systems. On the other hand, we construct a family of quantum hidden Markov networks, inspired by fault-tolerant MBQC, that sustain long-range CMI at low temperatures.
%
%This work leads to several new developments in various quantum information and many-body physics questions. First, by connecting CMI to teleportation, we show that long-range teleportation cannot occur at high temperatures in the wide range of setup described previously. Meanwhile, by showing the existence of long-range CMI at low temperatures, we establish an information-theoretic phase transition that is independent of the thermodynamic phase transitions. This is parallel to the recent interests in identifying space-time fault-tolerance as a new phase of matter \cite{negari2024spacetime}.
%
%From a perspective of condensed matter, divergent Markov length is a necessary condition for a system to undergo a mixed-state phase transition \cite{sang2024stability}. Therefore, we show that all high-temperature classical and quantum hidden Markov networks previously described are in the trivial phase. 
Our result also has implications for quantum state preparation (see Appendix \ref{recovery}). 
Lastly, it has been shown that for the measurement outcome distribution from any state, if the distribution has a finite Markov length, then it can be efficiently represented by a neural quantum state (See Appendix \ref{neural}) \cite{yang2024can} . Since measurements can be considered as a dephasing channel, our results imply that all high-temperature hidden Markov networks admit efficient neural quantum state representations.

\begin{acknowledgments}
Y.Z. and S.G. would like to thank Angela Capel, Tomotaka Kuwahara, Dominik Wild, Shengqi Sang, and Ewin Tang for useful discussions. Y.Z. and S.G. acknowledge support from NSF QuSEC-TAQS OSI 2326767.
\end{acknowledgments}

\begin{widetext}
\appendix
\section{Generalized Cluster Expansion Under Local Unital, Commutation-Preserving Channels}\label{high_temp_appn}
In this section, we prove the finite Markov length of any high-temperature quantum hidden Markov network under unital, commutation-preserving channels.

We consider an operator basis $\{P_\alpha\}$ of the $q$ dimensional local Hilbert space, where each operator has the operator norm $\norm{P_a} \le 1$. Recall that the operator norm of a matrix $A$ is defined as follows:
\begin{equation}
    \norm{A} = \max_{x} \frac{\norm {Ax}}{\norm {x}}
\end{equation}
Where $x$ denotes a vector with a norm. In the case of finite-dimensional Hermitian matrices, which are what we consider, the operator norm coincides with the Schatten-$\infty$ norm, defined as the largest magnitude of the eigenvalues.

Without loss of generality, we also assume that $\{P_\alpha\}$ contains the identity operator $I$ and $\text{Tr}[P_\alpha]=0$ for all $P_\alpha \neq I$. Finally, we demand that $\{P_\alpha\}$ form a projective representation of a group. In other words, any $P_a P_b$ still belongs to $\{P_\alpha\}$, up to a global phase.

We will consider the commuting Hamiltonian $H=\sum_a \lambda_a h_a$ where all $h_a$ commute mutually. We demand that $h_a$ is a tensor product of the operators in $\{P_\alpha\}$. An example is that $\{P_\alpha\}$ is the set of Pauli operators and $H$ is a stabilizer Hamiltonian. We now state our result formally below.

\begin{theorem}
\label{thm:decay_cmi}
Given a Hamiltonian $H=\sum_a \lambda_a h_a$ where each $h_a$ is a tensor product of operators in $\{P_\alpha\}$,  all $h_a$ commute mutually, and $|\lambda_a| \le 1, \forall a$. Suppose that every site supports at most $\mathfrak{d}$ terms in $H$. Let $\rho = \frac{1}{Z}e^{-\beta H}$ be the Gibbs state at temperature $\beta$. Consider arbitrary subsystems $ABC$ where $A$ and $C$ are separated (not necessarily by $B$). Let $\mathcal{E} = \otimes_i \mathcal{E}_i $ be a product of single-site unital channels $\mathcal{E}_i$ acting on $B$. We demand that $\{\mathcal{E}_i\}$ is commutation-preserving.

When $\beta \le \beta_c = 1/(e (\mathfrak{d}+1)(1+e(\mathfrak{d}-1)))$, we have the following decay of CMI in $\mathcal{E}[\rho]$:
    \begin{equation}
        I(A:C|\mathcal{E}[B]) \le c \, \min(|\partial A|, |\partial C|) e^{-\frac{d_{AC}}{\xi}}
    \end{equation}
Where $c$ is an $O(1)$ constant, $|\partial A|$ denotes the boundary size of $A$, defined as the number of terms in $H$ that are supported both inside and outside $A$. $|\partial C|$ is similarly defined. $d_{AC}$ is the distance between $A$ and $C$, defined as the minimum weight of any connected cluster that connects $A$ and $C$. $\xi = O(\frac{\beta}{\beta_C})$ gives the Markov length.
\end{theorem}
Note that we restrict the action of the channel to $B$ only. This is because applying local channels to $A$ and $C$ can only decrease the CMI due to the data processing inequality.

\subsection{Preliminaries}
We first introduce the definition of clusters. We consider a set of $n$ spins, each spin supporting a $d$ dimensional local Hilbert space. We consider a Hamiltonian $H = \sum_{a \in M} \lambda_a h_a$ acting on the $d^n$ dimensional Hilbert space, where $a$ labels the Hamiltonian term and $M$ denotes the collection of all Hamiltonian terms. $\lambda_a$ are coefficients satisfying $|\lambda_a| \le 1$ and $h_a$ are Hermitian matrices with bounded operator norms $\norm{h_a} \le 1$ (we will always use the operator norm in this manuscript, unless otherwise noted).

We now define a cluster. A cluster $\mathbf{W}$ is a multiset of Hamiltonian terms. Specifically, $\mathbf{W}$ is a set of tuples of the form $(a, \mu(a))$, where $a \in M$ labels the Hamiltonian term and $\mu(a)$ is a positive integer that labels its multiplicity. We define the weight of $\mathbf{W}$ as $|\mathbf{W}|=\sum_{a} \mu(a)$. We also define $\mathbf{W}!=\prod_a \mu(a)!$. We visualize an example of a cluster in Fig. \ref{fig:cluster}(a). In this example, the cluster has one element supported on qubit 1, 2 and one element supported on qubit 2, 3, 4, 5.  It also has two elements supported on qubit 4, 6, thus a multiplicity of two for this element.

\begin{figure*}
\includegraphics[width=\linewidth]{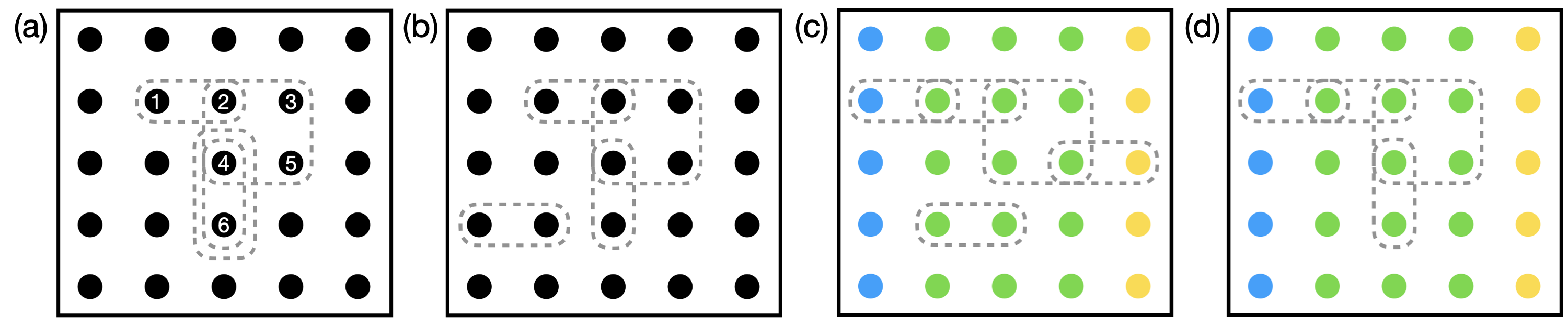}
\caption{\label{fig:cluster}(a) An example of a cluster with one element supported on qubit 1, 2, one element supported on qubit 2, 3, 4, 5, and two elements supported on qubit 4, 6. (b) An example of a disconnected cluster. (c) An example of a cluster connecting $\textcolor{Cerulean}{A}$ and $\textcolor{Goldenrod}{C}$ but is disconnected. (d) An example of a connected cluster but does not connect $\textcolor{Cerulean}{A}$ and $\textcolor{Goldenrod}{C}$.}
\end{figure*}

For any matrix $Q$ that depends on $\lambda_a$, we can write its multivariate Taylor expansion in terms of $\lambda_a$. 
\begin{align}
    Q = \sum_{\mathbf{W}} \frac{\lambda^{\mathbf{W}}}{\mathbf{W}!} \mathcal{D}_{\mathbf{W}} Q
\end{align}
Where we define $\lambda^{\mathbf{W}} = \prod_{a \in \mathbf{W}} \lambda_{a}^{\mu(a)}$ and $\mathcal{D}_{\mathbf{W}}=\prod_{a \in \mathbf{W}}(\frac{\partial}{\partial \lambda_a})^{\mu(a)} \Bigr|_{\lambda=0}$.

\subsection{Generalized Cluster Expansion}
In this section, we show that generalized cluster expansion \cite{kuwahara2020clustering} works even under local unital channels. Let the Hamiltonian be $H = \sum_a \lambda_a h_a$, where for all $a$ we have $|\lambda_a| \le 1$ and $\norm{h_a} \le 1$. We will consider channels of the form: $\mathcal{E}=\mathcal{E}_1 \otimes \mathcal{E}_2 \otimes \ldots \otimes \mathcal{E}_n$, where $\mathcal{E}_i$ are local channels acting on site $i$.

\begin{lemma}
Let $\mathcal{E}=\mathcal{E}_1 \otimes \mathcal{E}_2 \otimes \ldots \otimes \mathcal{E}_n$ be a product of local channels, then the CMI of $\mathcal{E}[\rho]$ can be written as
\begin{align}
    I(A:C|\mathcal{E}[B]) &= -\text{Tr}\,\left[\mathcal{E}[\rho_{ABC} ] H(A:C|\mathcal{E}[B])\right], \\
    H(A:C|\mathcal{E}[B]) &=  \log(\mathcal{E}[\tilde{\rho}_{AB}]) + \log(\mathcal{E}[\tilde{\rho}_{BC}]) - \log(\mathcal{E}[\tilde{\rho}_{B}])  - \log(\mathcal{E}[\tilde{\rho}_{ABC}])
\end{align}
Where $\tilde{\rho_L} = \text{Tr}_{L_c} (e^{-\beta H}) \otimes I_{L_c}$ denotes the embedding of the unnormalized reduced density matrix into the full Hilbert space. $L_c$ denotes the complement of region $L$ and $I_{L_c}$ is the identity operator acting on $L_c$.
\end{lemma}
We also define $\tilde{\rho}=e^{- \beta H}$ to be the unnormalized density matrix. An important observation is that tracing out is equivalent to the full depolarizing channel.
\begin{proposition}
\label{prop:trace_channel}
The partially traced state $\tilde{\rho}_L$ can be equivalently written as
    \begin{equation}
        \tilde{\rho}_L = \mathcal{E}^{Tr}_{L_c}[\tilde{\rho}]
    \end{equation}
    Where $\mathcal{E}^{Tr}_{L_c} = \otimes_{i \in L_c} \mathcal{E}^{Tr}_{i}$ is the product of all single-qubit complete depolarization channels in $L_c$. Further, Consider any set of commutation-preserving channels $\{\mathcal{E}_i\}$, where each $i$ is in region $L$. Suppose we include another set of depolarization channels $\{\mathcal{E}_{j}^{Tr}\}$ in $L_c$, then the set $\{\mathcal{E}_i\} \cup \{\mathcal{E}_{j}^{Tr}\}$ is also commutation-preserving.
\end{proposition}
\begin{proof}
The relation between partial trace and depolarization is trivial, so it remains to show that $\{\mathcal{E}_i\} \cup \{\mathcal{E}_{j}^{Tr}\}$ is commutation-preserving. 

Consider $O_m$ in Definition \ref{def:comm_preserv}. Since each $h_a$ is a product of $\{P_\alpha\}$ which forms a projective representation of a group, $O_m$ is also a product of $\{P_\alpha\}$. Next, we use the fact that $\mathcal{E}_{j}^{Tr}[P_\alpha]=0, \forall P_\alpha \neq I$ because $P_\alpha$ are traceless. Therefore, $\mathcal{E}_{j}^{Tr}[O_m]$ is either zero or $O_m$. Thus, the commutation preservation of $\{\mathcal{E}_i\} \cup \{\mathcal{E}_{j}^{Tr}\}$ follows from the commutation preservation of $\{\mathcal{E}_i\}$.
\end{proof}
As a direct corollary, Let  $\mathcal{E}=\mathcal{E}_1 \otimes \mathcal{E}_2 \otimes \ldots \otimes \mathcal{E}_n$ be a product of local channels, then
\begin{align}
    \mathcal{E}[\tilde{\rho}_L]  = (\mathcal{E}^{Tr}_{L_c} \circ \mathcal{E})[\tilde{\rho}]
\end{align}
Where $\mathcal{E}^{Tr}_{L_c} \circ \mathcal{E}$ is also a product of local channels that are all depolarization channels in $L_c$. Therefore, we use the perspective of applying depolarizing channels rather than tracing out qubits throughout the discussion.

Note that the definition of $H(A:C|\mathcal{E}[B])$ here differs slightly from the definition in the main text (Eq. (\ref{eq:mcmi_main_text})), where the normalized density matrix is used. Nevertheless, the two definitions are equivalent as any choice of normalization will cancel out after taking the logarithm. As stated in the main text, the operator norm $\norm{H(A:C|\mathcal{E}[B])}$ upper bounds CMI.
\begin{proposition}
\label{prop:cmi_norm_bound}
    CMI is upper-bounded by the operator norm of $H(A:C|\mathcal{E}[B])$.
    \begin{equation}
        I(A:C|\mathcal{E}[B])  \le \norm{H(A:C|\mathcal{E}[B])} 
    \end{equation}
\end{proposition}
We will use the operator norm throughout this paper, unless otherwise noted. We cluster expand $\mathcal{E}[\tilde{\rho}_L]$ and $\log(\mathcal{E}[\tilde{\rho}_L])$
\begin{align}
\mathcal{E}[\tilde{\rho}_L] &= \sum_{\mathbf{W}} \frac{\lambda_{\mathbf{W}}}{\mathbf{W}!}\mathcal{D}_\mathbf{W} \mathcal{E}[\tilde{\rho}_L]\\
\log(\mathcal{E}[\tilde{\rho}_L]) &= \sum_{\mathbf{W}} \frac{\lambda_{\mathbf{W}}}{\mathbf{W}!}\mathcal{D}_\mathbf{W} \log(\mathcal{E}[\tilde{\rho}_L])
\end{align}

\subsection{Connected Clusters}
The cluster expansions of $\mathcal{E}[\tilde{\rho}_L]$ and $\log(\mathcal{E}[\tilde{\rho}_L])$ admit significant simplifications by considering the connectedness of clusters. 

We take the set $M$ and construct a simple, unidirectional graph termed the dual interaction graph $\mathfrak{B}$ as follows. $\mathfrak{B}$ contains $|M|$ vertices corresponding to elements of $M$. Two vertices $a$ and $b$ of $\mathfrak{B}$ are connected if and only if $h_a$ and $h_b$ have overlapping support. Let the degree of $\mathfrak{B}$ be $\mathfrak{d}$. In other words, $\mathfrak{d}$ is the maximal number of Hamiltonian terms any site supports.

We say that a cluster $\mathbf{W}$ is \emph{connected} if the corresponding dual interaction graph is connected. In other words, $\mathbf{W}$ cannot be decomposed into a union of two clusters $\mathbf{W}_1 \cup \mathbf{W}_2$ such that $\mathbf{W}_1$ and $\mathbf{W}_2$ have disjoint support. See Fig. \ref{fig:cluster}(b) for an example of a disconnected cluster. We define $\mathcal{G}$ as the set of all connected clusters and $\mathcal{G}_m$ as the set of connected clusters with weight $m$.

When $\mathbf{W}$ is disconnected, $\mathcal{D}_\mathbf{W} \mathcal{E}[\tilde{\rho}]$ and $\mathcal{D}_\mathbf{W} \log(\mathcal{E}[\tilde{\rho}])$ admit significant simplifications.
\begin{lemma}
\label{lem:conn_cluster}
Let $\mathcal{E} = \mathcal{E}_1 \otimes \mathcal{E}_2 \otimes \ldots \otimes \mathcal{E}_n$ be a product of local channels. Suppose $\mathbf{W}$ can be decomposed into two disconnected clusters $\mathbf{W}_1$ and $\mathbf{W}_2$. We decompose $\mathcal{E} = \mathcal{E}_{\mathbf{W}_1} \otimes \mathcal{E}_{\mathbf{W}_2} \otimes \mathcal{E}_{\bar{\mathbf{W}}}$, where $\mathcal{E}_{\mathbf{W}_1}$ includes channels acting on the support of $\mathbf{W}_1$, $\mathcal{E}_{\mathbf{W}_2}$ includes channels acting on the support of $\mathbf{W}_2$, and $\mathcal{E}_{\bar{\mathbf{W}}}$ includes channels acting on the remaining system. Let $\mathcal{E}^{({\mathbf{W}_1})} = \mathcal{E}_{\mathbf{W}_1} \otimes I_{\mathbf{W_2 \cup \bar{W}}}$ be the embedding of $\mathcal{E}_{\mathbf{W}_1}$ into the global Hilbert space by tensoring with the identity channel on the remaining system. $\mathcal{E}^{(\mathbf{W_2})}$ and $\mathcal{E}^{(\bar{\mathbf{W}})}$ are similarly defined.

Under the above notations, we have
\begin{align}
    \mathcal{D}_\mathbf{W} \mathcal{E}[\tilde{\rho}]= \mathcal{D}_{\mathbf{W}_1}  \mathcal{E}^{({\mathbf{W}_1})}[\tilde{\rho}] \times \mathcal{D}_\mathbf{W_2} \mathcal{E}^{(\mathbf{W_2})}[\tilde{\rho}] \times \mathcal{E}^{(\bar{\mathbf{W}})}[I]
\end{align}
Where $\times$ denotes matrix multiplication. Furthermore, when $\mathcal{E}$ is a product of unital channels, $\mathcal{D}_\mathbf{W} \mathcal{E}[\tilde{\rho}]=  \mathcal{D}_{\mathbf{W}_1}  \mathcal{E}^{({\mathbf{W}_1})}[\tilde{\rho}] \times \mathcal{D}_\mathbf{W_2} \mathcal{E}^{(\mathbf{W_2})}[\tilde{\rho}]$.
\end{lemma}
\begin{proof}
Since $\mathbf{W}$ can be decomposed into two disconnected clusters $\mathbf{W}_1$ and $\mathbf{W}_2$. We have
    \begin{align}
        \mathcal{D}_\mathbf{W} \mathcal{E}[\tilde{\rho}] = \prod_{i \in  \mathbf{W}_1} (\frac{\partial }{\partial \lambda_i})^{\mu_i}  \prod_{j \in  \mathbf{W}_2} (\frac{\partial }{\partial \lambda_j})^{\mu_j}  \mathcal{E}[\tilde{\rho}]\Bigr|_{\lambda=0}
    \end{align}
 We first set the coefficients $\lambda$ that are not in $\mathbf{W}_1$ or $\mathbf{W}_2$ to zero. Because $\mathbf{W}_1$ and $\mathbf{W}_2$ are disconnected, the resulting density matrix factorizes.
 \begin{align}
        \mathcal{D}_\mathbf{W} \mathcal{E}[\tilde{\rho}] &= \prod_{i \in  \mathbf{W}_1} (\frac{\partial }{\partial \lambda_i})^{\mu_i}  \prod_{j \in  \mathbf{W}_2} (\frac{\partial }{\partial \lambda_j})^{\mu_j}\mathcal{E}\left[\rho^{(\mathbf{W}_1)}\times \rho^{(\mathbf{W}_2)}\right] \Bigr|_{\lambda_i=\lambda_j=0} \\
        \rho^{(\mathbf{W}_1)} &= e^{-\beta\sum_{a \in \mathbf{W}_1} \lambda_a h_a}  \label{eq:quantum_w1_def}  \\
        \rho^{(\mathbf{W}_2)} &= e^{-\beta\sum_{a \in \mathbf{W}_2} \lambda_a h_a}  \label{eq:quantum_w2_def} 
    \end{align}
    Where $\rho^{(\mathbf{W}_1)}$ and $\rho^{(\mathbf{W}_2)}$ denote the density matrix after setting all coefficients not in $\mathbf{W}_1$ or $\mathbf{W}_2$ to zero. Under the factorization $\mathcal{E} = \mathcal{E}_{\mathbf{W}_1} \otimes \mathcal{E}_{\mathbf{W}_2} \otimes \mathcal{E}_{\bar{\mathbf{W}}}$, we have
    \begin{equation}
  \mathcal{E}\left[\rho^{(\mathbf{W}_1)}\times \rho^{(\mathbf{W}_2)}\right] =  \mathcal{E}^{\mathbf{W}_1}\left[\rho^{(\mathbf{W}_1)}\right] \times \mathcal{E}^{\mathbf{W}_2}\left[\rho^{(\mathbf{W}_2)}\right] \times \mathcal{E}^{\bar{\mathbf{W}}}\left[I\right] 
    \end{equation}
And one can see that only $\rho_{\mathbf{W}_1}$ depends on $\mathbf{W}_1$ and only $\rho_{\mathbf{W}_2}$ depends on $\mathbf{W}_2$. Therefore,
\begin{align}
    \mathcal{D}_\mathbf{W} \mathcal{E}[\tilde{\rho}] &= \mathcal{D}_{\mathbf{W}_1}\mathcal{E}^{\mathbf{W}_1}\left[\rho^{(\mathbf{W}_1)}\right] \times \mathcal{D}_{\mathbf{W}_2}\mathcal{E}^{\mathbf{W}_2}\left[\rho^{(\mathbf{W}_2)}\right] \times \mathcal{D}_{\bar{\mathbf{W}}}\mathcal{E}^{\bar{\mathbf{W}}}\left[I\right] \\
    &= \mathcal{D}_{\mathbf{W}_1}  \mathcal{E}^{({\mathbf{W}_1})}[\tilde{\rho}] \times \mathcal{D}_\mathbf{W_2} \mathcal{E}^{(\mathbf{W_2})}[\tilde{\rho}] \times \mathcal{E}^{(\bar{\mathbf{W}})}[I]
\end{align}
Lastly, when $\mathcal{E}$ is a product of unital channels, $\mathcal{E}^{(\bar{\mathbf{W}})}[I]=I$, and the result trivially follows.
\end{proof}

\begin{lemma}
\label{lem:conn_log_cluster}
Let $\mathcal{E} = \mathcal{E}_1 \otimes \mathcal{E}_2 \otimes \ldots \otimes \mathcal{E}_n$ be a product of local channels. When $\mathbf{W}$ is disconnected, $\mathcal{D}_\mathbf{W} \log(\mathcal{E}[\rho])=0$
\end{lemma}
\begin{proof}
    Suppose $\mathbf{W}$ can be decomposed into two disconnected clusters $\mathbf{W}_1$ and $\mathbf{W}_2$. We have
    \begin{align}
        \mathcal{D}_\mathbf{W} \log(\mathcal{E}[\tilde{\rho}]) = \prod_{i \in  \mathbf{W}_1} (\frac{\partial }{\partial \lambda_i})^{\mu_i}  \prod_{j \in  \mathbf{W}_2} (\frac{\partial }{\partial \lambda_j})^{\mu_j}  \log(\mathcal{E}[\tilde{\rho}]) \Bigr|_{\lambda=0}
    \end{align}
 We first set the coefficients $\lambda$ that are not in $\mathbf{W}_1$ or $\mathbf{W}_2$ to zero. Because $\mathbf{W}_1$ and $\mathbf{W}_2$ are disconnected, the resulting density matrix factorizes.
    \begin{align}
        \mathcal{D}_\mathbf{W} \log(\mathcal{E}[\tilde{\rho}]) = \prod_{i \in  \mathbf{W}_1} (\frac{\partial }{\partial \lambda_i})^{\mu_i}  \prod_{j \in  \mathbf{W}_2} (\frac{\partial }{\partial \lambda_j})^{\mu_j}  \log(\mathcal{E}\left[\rho^{(\mathbf{W}_1)}\times \rho^{(\mathbf{W}_2)}\right]) \Bigr|_{\lambda_i=\lambda_j=0}
    \end{align}
Where $\rho^{(\mathbf{W}_1)}$ and $\rho^{(\mathbf{W}_2)}$ are defined in Eq. (\ref{eq:quantum_w1_def}, \ref{eq:quantum_w2_def}). Since $\mathcal{E}$ is a product of local channels, we separate it into three parts: $\mathcal{E} = \mathcal{E}_{\mathbf{W}_1} \otimes \mathcal{E}_{\mathbf{W}_2} \otimes \mathcal{E}_{\bar{\mathbf{W}}}$, where $\mathcal{E}_{\mathbf{W}_1}$ includes channels acting on the support of $\mathbf{W}_1$, $\mathcal{E}_{\mathbf{W}_2}$ includes channels acting on the support of $\mathbf{W}_2$, and $\mathcal{E}_{\bar{\mathbf{W}}}$ includes channels acting on the remaining system. Under such factorization, the logarithm becomes
\begin{align}
    \log\left(\mathcal{E}\left[\rho^{(\mathbf{W}_1)}\times \rho^{(\mathbf{W}_2)}\right]\right) &= \log\left(\mathcal{E}^{(\mathbf{W}_1)}\left[\rho^{(\mathbf{W}_1)}\right]\times \mathcal{E}^{(\mathbf{W}_2)}\left[\rho^{(\mathbf{W}_2)}\right] \times \mathcal{E}^{(\bar{\mathbf{W}})}\left[I\right]\right) \\
    &= \log\left(\mathcal{E}^{(\mathbf{W}_1)}\left[\rho^{(\mathbf{W}_1)}\right] \right) + \log \left( \mathcal{E}^{(\mathbf{W}_2)}\left[\rho^{(\mathbf{W}_2)}\right] \right) + \log \left( \mathcal{E}^{(\bar{\mathbf{W}})}\left[I\right]\right)
\end{align}
Where in the second line we use the fact that $\mathcal{E}^{(\mathbf{W}_1)}\left[\rho^{(\mathbf{W}_1)}\right]$, $\mathcal{E}^{(\mathbf{W}_2)}\left[\rho^{(\mathbf{W}_2)}\right]$, and $\mathcal{E}^{(\bar{\mathbf{W}})}\left[I\right]$ have disjoint support to factorize the matrix logarithm into the sum of three terms. After expressing the matrix logarithm as a sum of three terms, notice that the first terms only depend on $\mathbf{W}_1$, the second term only depends on $\mathbf{W}_2$, and the last term is independent of $\mathbf{W}_1$ or $\mathbf{W}_2$. Therefore, all three terms become zero after taking the cluster derivative of $\mathbf{W} = \mathbf{W}_1 \cup \mathbf{W}_2$.
\end{proof}

We will upper-bound CMI using the cluster expansion of $H(A:C|\mathcal{E}[B])$.
\begin{equation} \label{eq:mcmi_cluster_expand}
    H(A:C|\mathcal{E}[B]) = \sum_{\mathbf{W}} \frac{\lambda_{\mathbf{W}}}{\mathbf{W}!}\mathcal{D}_\mathbf{W} H(A:C|\mathcal{E}[B])
\end{equation}
A crucial observation is that the above cluster expansion is non-trivial only when $\mathbf{W}$ is a connected cluster connecting $A$ and $C$. Fig. \ref{fig:cluster}(c) illustrates an example where a cluster connects $A$ and $C$ but is disconnected, and Fig. \ref{fig:cluster}(d) illustrates an example where a cluster is connected but does not connect $A$ and $C$.
\begin{lemma}\label{lem:conn_AC}
    $\mathcal{D}_\mathbf{W} H(A:C|\mathcal{E}[B])$ is non-trivial only when $\mathbf{W}$ is a connected cluster connecting $A$ and $C$.
\end{lemma}
\begin{proof}
    If $\mathbf{W}$ is disconnected, $\mathcal{D}_\mathbf{W} H(A:C|\mathcal{E}[B]) = 0$ follows from Lemma \ref{lem:conn_log_cluster}. Without loss of generality, suppose $\mathbf{W}$ is connected but does not connect to $C$, then we first set all coefficients of terms supported on $C$ to zero. We first consider $\mathcal{D}_\mathbf{W}(\log(\mathcal{E}[\tilde{\rho}_{AB}])$ and $\mathcal{D}_\mathbf{W}(\log(\mathcal{E}[\tilde{\rho}_{ABC}])$.
\begin{align}
        \mathcal{D}_\mathbf{W} \log(\mathcal{E}[\tilde{\rho}_{AB}]) = \mathcal{D}_\mathbf{W} \log((\mathcal{E} \circ \mathcal{E}^{Tr}_{C}) [\tilde{\rho}_{ABC}]) \\
        = \mathcal{D}_\mathbf{W} \log((\mathcal{E} \circ \mathcal{E}^{Tr}_{C}) [\rho^{(\mathbf{W})}]) = \mathcal{D}_\mathbf{W} \log(\mathcal{E} [\rho^{(\mathbf{W})}])
\end{align}
Where $\rho^{(\mathbf{W})}=e^{-\beta\sum_{a \in \mathbf{W}} \lambda_a h_a}$ denotes the density matrix after setting all the coefficients not in $\mathbf{W}$ to zero. In the last line we use the fact that $\mathcal{E}^{Tr}_{C}$ is unital. Similarly,
\begin{align}
        \mathcal{D}_\mathbf{W} \log(\mathcal{E}[\tilde{\rho}_{ABC}]) = \mathcal{D}_\mathbf{W} \log(\mathcal{E} [\rho^{(\mathbf{W})}])
\end{align}
Therefore, $\mathcal{D}_\mathbf{W}(\log(\mathcal{E}[\tilde{\rho}_{AB}]) - \log(\mathcal{E}[\tilde{\rho}_{ABC}]) ) = 0$. Next, we consider $\mathcal{D}_\mathbf{W}(\log(\mathcal{E}[\tilde{\rho}_{BC}])$ and $\mathcal{D}_\mathbf{W}(\log(\mathcal{E}[\tilde{\rho}_{B}])$.
\begin{align}
      \mathcal{D}_\mathbf{W} \log(\mathcal{E}[\tilde{\rho}_{BC}]) = \mathcal{D}_\mathbf{W} \log((\mathcal{E} \circ \mathcal{E}^{Tr}_{A}) [\tilde{\rho}_{ABC}]) \\
        = \mathcal{D}_\mathbf{W} \log((\mathcal{E} \circ \mathcal{E}^{Tr}_{A}) [\rho^{(\mathbf{W})}])
\end{align}
\begin{align}
      \mathcal{D}_\mathbf{W} \log(\mathcal{E}[\tilde{\rho}_{B}]) = \mathcal{D}_\mathbf{W} \log((\mathcal{E} \circ \mathcal{E}^{Tr}_{A}\circ \mathcal{E}^{Tr}_{C}) [\tilde{\rho}_{ABC}]) \\
        = \mathcal{D}_\mathbf{W} \log((\mathcal{E} \circ \mathcal{E}^{Tr}_{A}\circ \mathcal{E}^{Tr}_{C}) [\rho^{(\mathbf{W})}]) = \mathcal{D}_\mathbf{W} \log((\mathcal{E} \circ \mathcal{E}^{Tr}_{A}) [\rho^{(\mathbf{W})}])
\end{align}
Therefore, $\mathcal{D}_\mathbf{W}(\log(\mathcal{E}[\tilde{\rho}_{BC}]) - \log(\mathcal{E}[\tilde{\rho}_{B}]) ) = 0$. Together, we establish $\mathcal{D}_\mathbf{W} H(A:C|\mathcal{E}[B])=0$
\end{proof}

We will prove Theorem \ref{thm:decay_cmi} through two steps. First, we show that the number of connected clusters at weight $w$ grows as $a^w$, where $a$ is some $O(1)$ constant. Second, the magnitude of the cluster derivative is upper-bounded by $(c\beta)^{w}$, where $c$ is some other $O(1)$ constant, together, we will establish the convergence of the cluster expansion in Eq. (\ref{eq:mcmi_cluster_expand}) at $\beta \le 1/(ac)$.

First, we quote the result that the number of connected clusters grows at most exponentially.
\begin{lemma}
\label{lem:num_conn_cluster} (Proposition 3.6 of \cite{haah2022optimal}) The number of connected clusters supported on any site $a$ with weight $w$ is upper bounded by $e\mathfrak{d}(1 + e(\mathfrak{d} - 1))^{w-1}$, where $\mathfrak{d}$ denotes the degree of the interaction graph.   
\end{lemma}
The major technical part is to establish the exponential decay of the cluster derivative, which we show in the next subsection.

\subsection{Bounding Cluster Derivative Under Unital Channels}
In this subsection we upper-bound the magnitude of cluster derivatives, which is the main technical step of the proof. Recall that $\tilde{\rho}=e^{- \beta H}$ denotes the unnormalized density matrix. We show that the norm of any cluster derivatives with weight $m$ is upper-bounded by $(c\beta)^{m}$, where $c$ is some $O(1)$ constant. This is the step that was found to be flawed in Ref. \cite{kuwahara2020clustering}. To address it, we employ a different combinatorial estimate introduced in \cite{wild2023classical,haah2022optimal} and generalize their technique to density matrices under local channels.
\begin{lemma}
\label{lem:cluster_deriv_upper}
    Given any cluster $\mathbf{W}$ with order $|\mathbf{W}|=m$,  for any channel $\mathcal{E} = \otimes_{i} \mathcal{E}_i$ that is a product of unital, commutation-preserving channels, 
\begin{align}
    \frac{1}{\mathbf{W}!}\norm{\mathcal{D}_{\mathbf{W}} \log(\mathcal{E}[\tilde{\rho}])} \le (2e(\mathfrak{d}+1)\beta)^{m+1}
\end{align}
\end{lemma}

We begin by supplying some short lemmas which will be useful later. First, any unital channel is contractive in the sense that its action cannot increase the operator norm.
\begin{lemma}\label{lem:unital_channel_op_norm}
  (Theorem II.4 of \cite{perez2006contractivity})  Given any matrix $A$ such that $\norm{A} = 1$, then for any unital channel $\mathcal{E}$, $\norm{\mathcal{E}[A]} \le 1$.
\end{lemma}
With that, we can upper-bound the operator norm of $\mathcal{D}_\mathbf{V}\mathcal{E}^{(\mathbf{V})}[\tilde{\rho}]$, where $\mathcal{E}^{(\mathbf{V})}$, defined in Lemma \ref{lem:conn_cluster}, is a product of local unital channels supported on $\mathbf{V}$.
\begin{lemma}\label{lem:deriv_bound}
    Suppose $|\mathbf{V}|=k$, then for any $\mathcal{E}^{(\mathbf{V})}$ that is a product of local unital channels supported on $\mathbf{V}$, $\norm{\mathcal{D}_\mathbf{V}\mathcal{E}^{(\mathbf{V})}[\tilde{\rho}]} \le \beta^k$
\end{lemma}
\begin{proof}
First expand $\tilde{\rho} = \sum_{n=1}^{\infty}\frac{(-\beta)^n}{n!} H^n $
\begin{align}
    \norm{\mathcal{D}_\mathbf{V}\mathcal{E}^{(\mathbf{V})}[\tilde{\rho}]} = \norm{\mathcal{D}_\mathbf{V}\mathcal{E}^{(\mathbf{V})}\left[\sum_{n=1}^{\infty} \frac{(-\beta)^n}{n!} H^n\right]} \le \sum_{n=1}^{\infty} \frac{\beta^n}{n!} \norm{\mathcal{D}_\mathbf{V}\mathcal{E}^{(\mathbf{V})}[H^n]}
\end{align}
Suppose $\mathbf{V}$ contain terms $h^{(\mathbf{V})}_1, h^{(\mathbf{V})}_2, \ldots, h^{(\mathbf{V})}_k$, where one term can show up multiple times to account for its multiplicity. Only the term on the $k$-th order contributes, so
\begin{align}
    \norm{\mathcal{D}_\mathbf{V}\mathcal{E}^{(\mathbf{V})}[\tilde{\rho}]} \le \frac{\beta^k}{k!} \norm{  \mathcal{D}_\mathbf{V}\mathcal{E}^{(\mathbf{V})}[H^k]} = \frac{\beta^k}{k!} \sum_{\sigma \in S_k} \norm{ \mathcal{E}^{(\mathbf{V})}[h^{(V)}_{\sigma(1)}h^{(V)}_{\sigma(2)} \ldots h^{(V)}_{\sigma(k)}]}
\end{align}
Where $S_k$ is the symmetric group of order $k$ and $\sigma \in S_k$ permutes the label. Finally, we invoke Lemma \ref{lem:unital_channel_op_norm} to have $\norm{ \mathcal{E}^{(\mathbf{V})}[h^{(V)}_{\sigma(1)}h^{(V)}_{\sigma(2)} \ldots h^{(V)}_{\sigma(k)}]} \le \norm{h^{(V)}_{\sigma(1)}h^{(V)}_{\sigma(2)} \ldots h^{(V)}_{\sigma(k)}} \le 1$,  so in the end,
\begin{align}
    \norm{\mathcal{D}_\mathbf{V}\mathcal{E}^{(\mathbf{V})}[\tilde{\rho}]} \le \frac{\beta^k}{k!} \sum_{\sigma \in S_k} \norm{h^{(V)}_{\sigma(1)}h^{(V)}_{\sigma(2)} \ldots h^{(V)}_{\sigma(k)}} \le \beta^k
\end{align}
\end{proof}

To establish Lemma \ref{lem:cluster_deriv_upper}, we first connect the cluster derivative to the graph coloring problem and employ a combinatorial estimate introduced in \cite{haah2022optimal,wild2023classical} to upper-bound the magnitude of cluster derivatives. We will need to introduce the notion of \emph{graph partition} and \emph{cluster partition} which are essentially the same object but are different in terms of their redundancies.

We start by defining the simple \emph{interaction graph} $\text{Gra}(\mathbf{W})$ of a cluster $\mathbf{W}$. $\text{Gra}(\mathbf{W})$ contains $|\mathbf{W}|$ nodes labeled with elements of $\mathbf{W}$. Specifically, if $\mathbf{W} = \{(a, \mu(a))\}$, each node can be labeled by a tuple $(a, i)$ where $i$ takes integer value from one to $\mu(a)$. Two nodes are connected if their corresponding $h_a$ and $h_b$ have overlapping support. A \emph{graph partition} $B$ of $\mathbf{W}$ is defined as the graph partition of $\text{Gra}(\mathbf{W})$. We will mostly consider a special subset of graph partitions where each partition has connected induced subgraph. We denote $\text{PaC}(F)$ as the collection of all graph partitions of $F$ such that they partition $F$ into connected induced subgraph.

On the other hand, we define a \emph{cluster partition} of $\mathbf{W}$ as a multiset $P=\{(\mathbf{W}_i,\mu(\mathbf{W}_i))\}$ such that the multiset union gives $\mathbf{W}$. We let $|P|=\sum_i \mu(\mathbf{W}_i)$ and $P! = \prod_i \mu(\mathbf{W}_i)!$. One can see that each graph partition $B$ of $\mathbf{W}$ corresponds to a cluster partition $P$ by simply "forgetting" the $i$ in the node label $(a, i)$. On the other hand, for each cluster partition $P$, there are $\frac{\mathbf{W}!}{P! \prod_{i}\mathbf{W}_i!}$ graph partitions that correspond to $P$. Similarly, we denote $\text{PaC}(\mathbf{W})$ as the collection of all cluster partitions of $\mathbf{W}$ into connected clusters. 

In essence, a graph partition is similar to a cluster partition, but when $\mu(\mathbf{W}_i) \ge 1$, then different $\mathbf{W}_i$ in the multiset are treated as distinguishable by assigning labels to each one of them.

The interaction graph $\text{Gra}(P)$ of a cluster partition $P$ is defined as follows: $\text{Gra}(P)$ contains $|P|$ nodes corresponding to clusters, and two nodes are connected if and only if their corresponding clusters $\mathbf{W}_i$ and $\mathbf{W}_j$ are connected after taking the union $\mathbf{W}_i \cup \mathbf{W}_j$. The interaction graph $\text{Gra}(B)$ of a graph partition $B$ is defined similarly. For any graph $F$, let $\chi^*(n,F)$ denote the number of node colorings using exactly $n$ colors such that two connected nodes have different colors.

\begin{lemma}
\label{lem:graph_coloring}
Given any cluster $\mathbf{W}$ with order $|\mathcal{W}|=k$,  for any channel $\mathcal{E} = \otimes_{i} \mathcal{E}_i$ that is a product of unital, commutation-preserving channels, 
\begin{align}
        \norm{\mathcal{D}_{\mathbf{W}} \log(\mathcal{E}[\tilde{\rho}])} \le   \beta^{|\mathbf{W}|} \left( \sum_{B  \in \text{PaC}(\text{Gra}(\mathbf{W}))}   \sum_{n=1}^{|B|} \frac{(-1)^n}{n}\chi^*(n,\text{Gra}(B)) \right)
    \end{align}
\end{lemma}
 The above lemma is a direct generalization of Lemma 3.11 from \cite{haah2022optimal} to density matrices under local unital channels.
\begin{proof}
    We first cluster expand $\mathcal{E}[\tilde{\rho}]$ 
\begin{align}
     \mathcal{E}[\tilde{\rho}] =\mathcal{E}[I]+ \sum_{k=1}^{\infty}  \sum_{\mathbf{W}:|\mathbf{W}| =k} \frac{\lambda_{\mathbf{W}}}{\mathbf{W}!}\mathcal{D}_\mathbf{W}\mathcal{E}[\tilde{\rho}] = I+ \sum_{k=1}^{\infty}  \sum_{\mathbf{W}:|\mathbf{W}| =k} \frac{\lambda_{\mathbf{W}}}{\mathbf{W}!}\mathcal{D}_\mathbf{W}\mathcal{E}[\tilde{\rho}]
\end{align}
Where in the second equality we use the unital property of the channel to have $\mathcal{E}[I]=I$. In general, $\mathbf{W}$ may be disconnected. We use $P_{max}(\mathbf{W})$ to denote the maximally connected subset of $\mathbf{W}$, namely the minimal partition that separates $\mathbf{W}$ into connected subsets. Using Lemma \ref{lem:conn_cluster}, we have
\begin{align}
   \mathcal{E}[\tilde{\rho}] =  I+ \sum_{k=1}^{\infty}  \sum_{\mathbf{W}:|\mathbf{W}| =k} \prod_{\mathbf{V} \in P_{max}(\mathbf{W})} (\frac{\lambda_{\mathbf{V}}}{\mathbf{V}!}\mathcal{D}_\mathbf{V}\mathcal{E}^{(\mathbf{V})}[\tilde{\rho}])
\end{align}
Where $\mathcal{E}^{(\mathbf{V})}$ was defined in Lemma \ref{lem:conn_cluster}.

Next, we apply the matrix logarithm expansion $\log(I+A)=\sum_{n=1}^{\infty} \frac{(-1)^{n-1}}{n}A^n$ to expand $\log(\mathcal{E}[\tilde{\rho_L}])$.
\begin{align}
    \log(\mathcal{E}[\tilde{\rho}]) &= \sum_{n=1}^{\infty} \frac{(-1)^{n-1}}{n} (\sum_{k=1}^{\infty}   \sum_{\mathbf{W}:|\mathbf{W}| =k} \frac{\lambda_{\mathbf{W}}}{\mathbf{W}!}\mathcal{D}_\mathbf{W}\mathcal{E}[\tilde{\rho}])^n \\
    &=\sum_{n=1}^{\infty} \frac{(-1)^{n-1}}{n} (\sum_{k=1}^{\infty}   \sum_{\mathbf{W}:|\mathbf{W}| =k} \prod_{\mathbf{V} \in P_{max}(\mathbf{W})} (\frac{\lambda_{\mathbf{V}}}{\mathbf{V}!}\mathcal{D}_\mathbf{V}\mathcal{E}^{(\mathbf{V})}[\tilde{\rho}]))^n \label{eq:log_expand}
\end{align}
To estimate $\mathcal{D}_{\mathbf{W}} \log(\mathcal{E}[\tilde{\rho}])$, we need to reorganize the above equation into a cluster expansion of $\log(\mathcal{E}[\tilde{\rho}])$, formally shown below.
\begin{align}\label{eq:log_expand_reorder}
    \log(\mathcal{E}[\tilde{\rho}]) &= \sum_{k=1}^{\infty}  \sum_{\mathbf{W} \in \mathcal{G}_k} \sum_{P \, \text{partitioning} \, \mathbf{W}}  C(P) \prod_{\mathbf{V} \in P} (\frac{\lambda_{\mathbf{V}}}{\mathbf{V}!}\mathcal{D}_\mathbf{V}\mathcal{E}^{(\mathbf{V})}[\tilde{\rho}])
\end{align}
Where $C(P)$ are coefficients that we will match to Eq. (\ref{eq:log_expand}). Here we will need the commutation-preserving property of the channel. $\mathcal{D}_\mathbf{V}\tilde{\rho}$ contain terms generated by $h_a$ (see Definition \ref{def:comm_preserv}). Since $\mathcal{E}^{(\mathbf{V})}$ is commutation-preserving, different $\mathcal{D}_\mathbf{V}\mathcal{E}^{(\mathbf{V})}[\tilde{\rho}]$ commutes, so we do not care about the ordering in the multiplication. When different $\mathcal{D}_\mathbf{V}\mathcal{E}^{(\mathbf{V})}[\tilde{\rho}]$ does not commute, mapping Eq. (\ref{eq:log_expand}) to Eq. (\ref{eq:log_expand_reorder}) requires shuffling the order of multiplying different $\mathcal{D}_\mathbf{V}\mathcal{E}^{(\mathbf{V})}[\tilde{\rho}]$. So far this seems to destroy the convergence of the generalized cluster expansion. However, we do not face this issue as long as $E^{(\mathbf{V})}$ are commutation-preserving.

As a reminder, $\mathcal{G}_m$ denotes the set of connected clusters with weight $k$. Note that we only sum over connected clusters since we know from Lemma \ref{lem:conn_log_cluster} that disconnected clusters do not contribute.

We will show how to reorganize Eq. (\ref{eq:log_expand}) into Eq. (\ref{eq:log_expand_reorder}). First, we expand the $n$-th power in Eq. (\ref{eq:log_expand}) as follows 
\begin{equation}
\begin{split}
    \log(\mathcal{E}[\tilde{\rho}]) 
    =\sum_{n=1}^{\infty} \frac{(-1)^{n-1}}{n}& \sum_{k_1=1}^{\infty}   \sum_{\mathbf{W}_1:|\mathbf{W}_1| =k_1}\sum_{k_2=1}^{\infty}   \sum_{\mathbf{W}_2:|\mathbf{W}_2| =k_2} \ldots \sum_{k_n=1}^{\infty}   \sum_{\mathbf{W}_n:|\mathbf{W}_n| =k_n} \\
    &\prod_{\mathbf{V}_{m,1} \in P_{max}(\mathbf{W}_1)} (\frac{\lambda_{\mathbf{V}_{m,1}}}{\mathbf{V}_{m,1}!}\mathcal{D}_{\mathbf{V}_{m,1}}\mathcal{E}^{(\mathbf{V}_{m,1})}[\tilde{\rho}]))  \\
\times &\prod_{\mathbf{V}_{m,2} \in P_{max}(\mathbf{W}_2)}(\frac{\lambda_{\mathbf{V}_{m,2}}}{\mathbf{V}_{m,2}!}\mathcal{D}_{\mathbf{V}_{m,2}}\mathcal{E}^{(\mathbf{V}_{m,2})}[\tilde{\rho}])) \\
\times \ldots  \times  &\prod_{\mathbf{V}_{m,n} \in P_{max}(\mathbf{W}_n)}(\frac{\lambda_{\mathbf{V}_{m,n}}}{\mathbf{V}_{m,n}!}\mathcal{D}_{\mathbf{V}_{m,n}}\mathcal{E}^{(\mathbf{V}_{m,n})}[\tilde{\rho}]))
\end{split}
\end{equation}
Where for each term in the $n$-th power we introduce a $k_n$, $\mathbf{W}_n$, and $\mathbf{V}_{m,n} \in P_{max}(\mathbf{W}_n)$. Next, we reorganize the summation over $k_n$, $\mathbf{W}_n$ according to the total cluster weight.
\begin{equation} \label{eq:log_expand_partition}
\begin{split}
    \log(\mathcal{E}[\tilde{\rho}]) 
    =\sum_{n=1}^{\infty} \frac{(-1)^{n-1}}{n}& \sum_{k=1}^{\infty}   \sum_{\{\mathbf{W}_i\}:\sum_{i=1}^{n}|\mathbf{W}_i| =k} \\
    &\prod_{\mathbf{V}_{m,1} \in P_{max}(\mathbf{W}_1)} (\frac{\lambda_{\mathbf{V}_{m,1}}}{\mathbf{V}_{m,1}!}\mathcal{D}_{\mathbf{V}_{m,1}}\mathcal{E}^{(\mathbf{V}_{m,1})}[\tilde{\rho}]))  \\
\times &\prod_{\mathbf{V}_{m,2} \in P_{max}(\mathbf{W}_2)}(\frac{\lambda_{\mathbf{V}_{m,2}}}{\mathbf{V}_{m,2}!}\mathcal{D}_{\mathbf{V}_{m,2}}\mathcal{E}^{(\mathbf{V}_{m,2})}[\tilde{\rho}])) \\
\times \ldots  \times  &\prod_{\mathbf{V}_{m,n} \in P_{max}(\mathbf{W}_n)}(\frac{\lambda_{\mathbf{V}_{m,n}}}{\mathbf{V}_{m,n}!}\mathcal{D}_{\mathbf{V}_{m,n}}\mathcal{E}^{(\mathbf{V}_{m,n})}[\tilde{\rho}]))
\end{split}
\end{equation}
Now each term corresponds to a cluster derivative of $\mathbf{W} = \cup_i \mathbf{W}_i$ with weight $k$, so the above equation formally corresponds to Eq. (\ref{eq:log_expand_reorder}), but with redundancies. In general, there could be different sets of $\{\mathbf{W}_i\}$ that give rise to the same $\mathbf{W}$ as they correspond to different cluster partitions $P$ of $\mathbf{W}$. We will count the number of cluster partitions and show that it is related to the graph coloring problem, which will help us determine the coefficient $C(P)$ in Eq. (\ref{eq:log_expand_reorder}).

Each term in Eq. (\ref{eq:log_expand_partition}) consists of cluster derivatives of a family of clusters $\{\mathbf{V}_{m,n}\}$ such that their union is a connected cluster $\mathbf{W}$ (we do not consider disconnected clusters because of Lemma \ref{lem:conn_log_cluster}, as shown in Eq. (\ref{eq:log_expand_reorder})). The set $\{\mathbf{V}_{m,n}\}$ corresponds to exactly the cluster partition $P$ in (\ref{eq:log_expand_reorder}).

Crucially, each $\mathbf{V}_{m,n}$ is connected and any $\mathbf{V}_{m,n}$ and $\mathbf{V}_{m',n}$ has to be disconnected. This is because because they belong to $P_{max}(\mathbf{W}_n)$, so $\mathbf{V}_{m,n}$ has to be connected by definition. Meanwhile, If $\mathbf{V}_{m,n}$ and $\mathbf{V}_{m',n}$ are connected, then one can construct a new partition that merges $\mathbf{V}_{m,n}$ and $\mathbf{V}_{m',n}$, thereby violating the condition of being a maximally connected subset. On the other hand, $\mathbf{V}_{m,n}$ and $\mathbf{V}_{m',n'}$ with $n \neq n'$ can in general be connected. The condition that $\mathbf{V}_{m,n}$ being connected implies that $P := \{\mathbf{V}_{m,n}\} \in \text{PaC}(\mathbf{W})$. In addition, the condition that $\mathbf{V}_{m,n}$ and $\mathbf{V}_{m',n}$ being disconnected correspond to exactly the graph coloring condition in $\text{Gra}(P)$: each value of $n$ is assigned a color and each node in $\text{Gra}(P)$, labeled by $\mathbf{V}_{m,n}$, is painted in the corresponding color. All nodes in the same color has to be mutually disconnected. 

On the other hand, there could be multiple graph colorings of $\text{Gra}(P)$ that correspond to the same $P$. This happens when any cluster multiplicity $\mu(\mathbf{W}_i) >1$, as permuting the color within the set of $\mathbf{W}_i$ gives rise to a different graph coloring but correspond to the same $P$. Therefore, each cluster partition $P$ corresponds to $P!$ graph colorings of $\text{Gra}(P)$. $C(P)$ then is the sum of all possible numbers of graph coloring of $\text{Gra}(P)$ with one, two, up to $|P|$ colors, divided by the $P!$ redundancy, and including the factor of $\frac{(-1)^n}{n}$ originating from the Taylor expansion of logarithm.
\begin{equation}
    C(P) = \frac{1}{P!} \sum_{n=1}^{|P|} \frac{(-1)^n}{n}\chi^*(n,\text{Gra}(P))
\end{equation}
Plugging it back to Eq. (\ref{eq:log_expand_reorder})
\begin{align}
    \log(\mathcal{E}[\tilde{\rho}]) &= \sum_{k=1}^{\infty}  \sum_{\mathbf{W}: |\mathbf{W}|=k} \sum_{P  \in \text{PaC}(\mathbf{W})}  \left(\frac{1}{P!} \sum_{n=1}^{|P|} \frac{(-1)^n}{n}\chi^*(n,\text{Gra}(P)\right) \prod_{\mathbf{V} \in P} (\frac{\lambda_{\mathbf{V}}}{\mathbf{V}!}\mathcal{D}_\mathbf{V}\mathcal{E}^{(\mathbf{V})}[\tilde{\rho}])
\end{align}
Summing over cluster partitions is equivalent to summing over graph partitions divided by the combinatorial factor of $\frac{\mathbf{W}!}{P! \prod_{m,n}\mathbf{V}_{m,n}!}$, so we have
\begin{align}
    \log(\mathcal{E}[\tilde{\rho}]) &= \sum_{k=1}^{\infty}  \sum_{\mathbf{W}: |\mathbf{W}|=k} \frac{1}{\mathbf{W}!} \left( \sum_{B  \in \text{PaC}(\text{Gra}(\mathbf{W}))}   \sum_{n=1}^{|B|} \frac{(-1)^n}{n}\chi^*(n,\text{Gra}(B))\right) \prod_{\mathbf{V} \in P} (\lambda_{\mathbf{V}}\mathcal{D}_\mathbf{V}\mathcal{E}^{(\mathbf{V})}[\tilde{\rho}])
\end{align}
Where the map from $B$ to $P$ is implicit (forgetting the label). Taking the cluster derivative selects the corresponding term
\begin{align}
    \mathcal{D}_{\mathbf{W}}\log(\mathcal{E}[\tilde{\rho}]) &= \left( \sum_{B  \in \text{PaC}(\text{Gra}(\mathbf{W}))}   \sum_{n=1}^{|B|} \frac{(-1)^n}{n}\chi^*(n,\text{Gra}(B))\right) \prod_{\mathbf{V} \in P} (\mathcal{D}_\mathbf{V}\mathcal{E}^{(\mathbf{V})}[\tilde{\rho}])
\end{align}
Now we can upper-bound the operator norm
\begin{align}
    \norm{\mathcal{D}_{\mathbf{W}}\log(\mathcal{E}[\tilde{\rho}])} &\le \left( \sum_{B  \in \text{PaC}(\text{Gra}(\mathbf{W}))}   \sum_{n=1}^{|B|} \frac{(-1)^n}{n}\chi^*(n,\text{Gra}(B)) \right) \norm{\prod_{\mathbf{V} \in P} (\mathcal{D}_\mathbf{V}\mathcal{E}^{(\mathbf{V})}[\tilde{\rho}])}\\
    &\le \left( \sum_{B  \in \text{PaC}(\text{Gra}(\mathbf{W}))}   \sum_{n=1}^{|B|} \frac{(-1)^n}{n}\chi^*(n,\text{Gra}(B)) \right) \prod_{\mathbf{V} \in P} \norm{\mathcal{D}_\mathbf{V}\mathcal{E}^{(\mathbf{V})}[\tilde{\rho}]}
\end{align}
Where in the second line we use the sub-multiplicity of the operator norm. It remains to upper-bound $\norm{\mathcal{D}_\mathbf{V}\mathcal{E}^{(\mathbf{V})}[\tilde{\rho}]}$. This is established in Lemma \ref{lem:deriv_bound}. Therefore, we have in the end
\begin{align}
    \norm{\mathcal{D}_{\mathbf{W}}\log(\mathcal{E}[\tilde{\rho}])} &\le \left( \sum_{B  \in \text{PaC}(\text{Gra}(\mathbf{W}))}   \sum_{n=1}^{|B|} \frac{(-1)^n}{n}\chi^*(n,\text{Gra}(B)) \right) \prod_{\mathbf{V} \in P} \beta^{|\mathbf{V}|} \\
    &\le \beta^{|\mathbf{W}|} \left( \sum_{B  \in \text{PaC}(\text{Gra}(\mathbf{W}))}   \sum_{n=1}^{|B|} \frac{(-1)^n}{n}\chi^*(n,\text{Gra}(B)) \right)
\end{align}
\end{proof}

The magnitude of the cluster derivative is upper-bounded by a sequence of combinatorial estimate of the graph coloring problem. We quote the result directly from \cite{haah2022optimal}. Ref. \cite{wild2023classical} also obtains a similar combinatorial estimate using the tutte polynomials.
\begin{lemma}
\label{lem:combinatorial_estimate}
For any cluster $\mathbf{W}$
\begin{align}
\sum_{B \in \text{PaC}(\text{Gra}(\mathbf{W}))} \left|\sum_{n=1}^{|B|} \frac{(-1)^n}{n}\chi(n,\text{Gra}(B)) \right| \le 2^{n-1} \tau(\text{Gra}(\mathbf{W})) \\
\le 2^{n-1} \prod_{a \in \text{node}(\text{Gra}(\mathbf{W}))} \text{deg}(a) \le \mathbf{W}! (2e(1+\mathfrak{d}))^{|\mathbf{W}|+1} 
\end{align}
Where $\text{Gra}(B)$ denotes the induced interaction graph of $B$, $\tau(G)$ denotes the number of spanning trees of $G$, and $\text{deg}(a)$ denotes the degree of vertex $a$.
\end{lemma}
The first inequality above comes from Lemma 3.12, the second inequality comes from the proof of Lemma 3.9, and the last inequality comes from the proof of Lemma Proposition 3.8 in \cite{haah2022optimal}.

\begin{proof}[Proof of Lemma \ref{lem:cluster_deriv_upper}]
First apply Lemma \ref{lem:graph_coloring}, then apply Lemma \ref{lem:combinatorial_estimate}.
\end{proof}

\subsection{Proof of Theorem \ref{thm:decay_cmi_informal} (formally Theorem \ref{thm:decay_cmi})}
By combining the exponential growth of the number of connected clusters (Lemma \ref{lem:num_conn_cluster}) and the exponential suppression of cluster derivatives' magnitude (Lemma \ref{lem:cluster_deriv_upper}), we are ready to establish the main theorem.

\begin{proof}[Proof of Theorem \ref{thm:decay_cmi}]
   We use Proposition \ref{prop:cmi_norm_bound} to relate CMI to $H(A:C|\mathcal{E}[B])$.
   \begin{equation}
       I(A:C|\mathcal{E}[B])  \le \norm{H(A:C|\mathcal{E}[B])} 
   \end{equation}
   To upper-bound $\norm{H(A:C|\mathcal{E}[B])}$, we take the cluster expansion of $H(A:C|\mathcal{E}[B])$.
    \begin{equation}
    H(A:C|\mathcal{E}[B]) = \sum_{\mathbf{W}} \frac{\lambda_{\mathbf{W}}}{\mathbf{W}!}\mathcal{D}_\mathbf{W} H(A:C|\mathcal{E}[B])
\end{equation}
Using Lemma \ref{lem:conn_AC}, the non-trivial clusters are those connecting $AC$ and are themselves connected. Denote the set of connected clusters that connect $AC$ by $\mathcal{G}^{AC}_m$. The minimal order of $m$ is $d_{AC}$
\begin{equation}
     H(A:C|\mathcal{E}[B]) = \sum_{m=d_{AC}}^{\infty} \sum_{\mathbf{W} \in \mathcal{G}^{AC}_m} \frac{\lambda_{\mathbf{W}}}{\mathbf{W}!}\mathcal{D}_\mathbf{W} H(A:C|\mathcal{E}[B])
\end{equation}
$H(A:C|\mathcal{E}[B])$ contains four terms that are all of the form $\mathcal{E}[\tilde{\rho}]$, where $\mathcal{E}$ is a product of local unital channels. This is because Proposition \ref{prop:trace_channel} shows the equivalence between the partial trace and the depolarization channel which is a product of local unital, commutation-preserving channels, and the composition of two products of local unital channels is also a product of local unital channels. Thus, we apply Lemma \ref{lem:cluster_deriv_upper} to upper-bound the norm of each term in the summation
\begin{align}
     \norm{H(A:C|\mathcal{E}[B])} &\le 4 \sum_{m=d_{AC}}^{\infty} \sum_{\mathbf{W} \in \mathcal{G}^{AC}_m} (2e(\mathfrak{d}+1)\beta)^{m+1}
\end{align}
Following Lemma \ref{lem:num_conn_cluster}, the number of clusters in $\mathcal{G}^{AC}_m$ is upper-bounded by $|\mathcal{G}^{AC}_m| \le c \, \text{min}(|\partial A|, |\partial C|) e\mathfrak{d}(1 + e(\mathfrak{d} - 1))^{m-1} $.\begin{align}
    \norm{H(A:C|\mathcal{E}[B])} &\le c \, \text{min}(|\partial A|, |\partial C|) \sum_{m=d_{AC}}^{\infty}  (1 + e(\mathfrak{d} - 1))^{m-1} (2e(\mathfrak{d}+1)\beta)^{m+1}
\end{align}
Where we absorb the factor of four into $c$. The above series converges absolutely when $\beta \le \beta_c = 1/(2 e (\mathfrak{d}+1)(1+e(\mathfrak{d}-1)))$. Furthermore, the leading order term is $O((\frac{\beta}{\beta_c})^{d_{AC}})$, so we have in the end,
\begin{align}
I(A:C|\mathcal{E}[B]) \le \norm{H(A:C|\mathcal{E}[B])} \le c \, \text{min}(|\partial A|, |\partial C|) e^{\frac{d_{AC}}{\xi}}
\end{align}
\end{proof}

\section{Finite Markov Length in Classical Gibbs Distribution at High Temperature}\label{classical_high_temp_appn}
In this section, we prove the finite Markov length in classical Gibbs distributions subject to arbitrary local transition matrices. 

We again consider an operator basis $\{D_\alpha\}$, where each $D_\alpha$ is diagonal in the computational basis, such that $\{D_\alpha\}$ generates all diagonal operators in the $q$ dimensional Hilbert space. We will assume that each operator has the operator norm $\norm{P_a} \le 1$. One example is $\{I,Z\}$ when the local Hilbert space is a qubit. We state our result formally below.
\begin{theorem}\label{thm:classical_decay_cmi}
    Given a Hamiltonian $H=\sum_a \lambda_a h_a$, where each $h_a$ is a tensor product of operators in $\{D_\alpha\}$ and $|\lambda_a| \le 1, \forall a$. Suppose every site supports at most $\mathfrak{d}$ terms in $H$. Let $\rho = \frac{1}{Z}e^{-\beta H}$ be the Gibbs distribution at temperature $\beta$. Consider arbitrary subsystems $XYZ$. Let $\mathcal{T} = \otimes_i \mathcal{T}_i$ be a product of transition matrix $\mathcal{T}_i$ each acting on site $i$ in $Y$. We demand that each $\mathcal{T}_i$ has non-zero matrix elements.

When $\beta \le \beta_c = {1/(e (\mathfrak{d}+1)(1+e(\mathfrak{d}-1)))}$, we have the following decay of CMI in $\mathcal{T}[\rho]$:
    \begin{equation}
        I(X:Z|\mathcal{T}[Y]) \le \, c \min(|\partial X|, |\partial Z|) e^{-\frac{d_{XZ}}{\xi}}
    \end{equation}
Where $c$ is an $O(1)$ constant, $|\partial X|$ denotes the boundary size of $X$, defined as the number of terms in $H$ that are supported both inside and outside $X$. $|\partial Z|$ is similarly defined. $d_{XZ}$ is the distance between $X$ and $Z$, defined as the minimal weight of any connected cluster that connects $X$ and $Z$. $\xi = O(\frac{\beta}{\beta_C})$ gives the Markov length.
\end{theorem}
The requirement that $\mathcal{T}_i$ has non-zero matrix element is technical but does not impose any practical restrictions. To improve the result to generic $\mathcal{T}_i$ with possibly zeros in matrix elements, one can simply mix $\mathcal{T}$ with an infinitesimally small amount of depolarizing channels. This operation only perturbs the CMI by an infinitesimal amount, so the upper bound still holds.

\subsection{Converting to Pinned Hamiltonian}
The proof is based on a post-selection trick and a variant of the generalized cluster expansion. First, we observe that the classical CMI can be written as the average of post-selected mutual information. The mutual information of a distribution $P(XZ)$ is defined as
\begin{equation}
    I(X,Z)=S(X) + S(Z) - S(XZ)
\end{equation}
Where $S(X)$, $S(Z)$, and $S(XZ)$ denote the Shannon entropy on $X$, $Z$, and $XZ$.
\begin{proposition}
   \label{prop:post_select}
    Given three random variables $X$, $W$, $Z$, their CMI has the following decomposition:
    \begin{equation}
        I(X:Z|W) = \sum_w P(W=w) I(X:Z|W=w)
    \end{equation}
    Where we use $w$ to label the values of $W$. $P(W=w)$ denotes the marginal probability that $W=w$, and $I(X:Z|W=w)$ denotes the mutual information of the distribution on $XZ$ conditioned on $W=w$.
\end{proposition}
We will set $W=\mathcal{T}[Y]$ and try to bound each post-selected mutual information, denoted as $I(X:Z|\mathcal{T}[Y]=y)$. The key step, stated in the next proposition, is to realize that classical hidden Markov networks become Gibbs distributions after post-selections.

\begin{figure*}
\includegraphics[width=\linewidth]{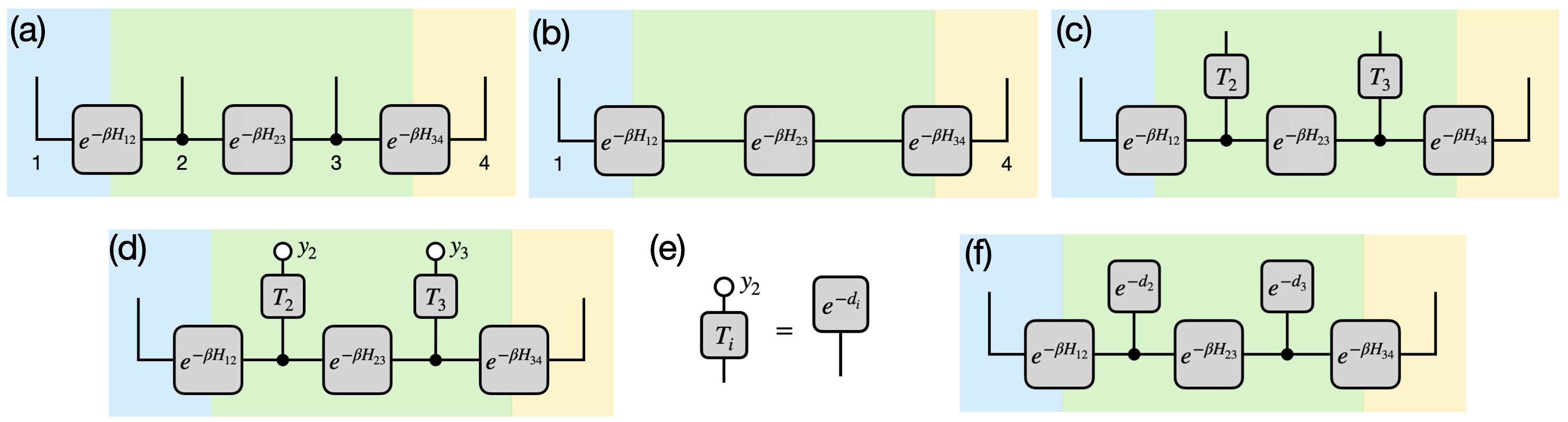}
\caption{\label{fig:factor_graph}(a) The factor graph of a Gibbs states with four bits in a line and nearest-neighbor interactions. We partition the four bits into $\textcolor{Cerulean}{X}$, $\textcolor{Green}{Y}$, and $\textcolor{Goldenrod}{Z}$.  (b) The marginal distribution on $\textcolor{Cerulean}{X}\textcolor{Goldenrod}{Z}$. (c) The Gibbs distribution is subject to two transition matrices $T_2$ and $T_3$ acting on the two bits in $\textcolor{Green}{Y}$. (d) The distribution on $\textcolor{Cerulean}{X}\textcolor{Goldenrod}{Z}$ conditioned on post-selecting on $y_2y_3$ in $\textcolor{Green}{Y}$. (e) Identifying the post-selected transition matrix as a new term in the factor graph. (f) The post-selected marginal distribution on $\textcolor{Cerulean}{X}\textcolor{Goldenrod}{Z}$ becomes the marginal distribution of a new Gibbs distribution, after applying the identification in (e).}
\end{figure*}

\begin{lemma}\label{lem:pinned_gibbs}
 Given a Hamiltonian $H=\sum_a \lambda_a h_a$, where each $h_a$ is diagonal in the computational basis. Let $\rho = \frac{1}{Z}e^{-\beta H}$ be the Gibbs distribution at temperature $\beta$. Let $\mathcal{T} = \otimes_i \mathcal{T}_i $ be a product of single-site transition matrix $\mathcal{T}_i$ acting on $Y$ such that all matrix elements are non-zero. $\mathcal{T}[\rho]$ denotes the distribution after applying the transition matrix. Let $\bar{Y}$ be the complementary region of $Y$. Let $\mathcal{T} [\rho]\bigr|_{y}$ be the distribution on $\bar{Y}$ conditioned on $Y=y$. In other words, the diagonal elements of $\mathcal{T} [\rho]\bigr|_{y}$ encode the probability distribution $P(\mathcal{T}[Y]=y,\bar{Y})$. Then,
 \begin{equation}
     \mathcal{T} [\rho]\bigr|_{y} \propto \text{Tr}_Y[e^{-H^{(y)}}]
 \end{equation}
 Where we have $H^{(y)} = \beta H + \sum_{i \in Y} d^{(i)}$. $d^{(i)}$ is some Hermitian matrix supported on site $i$ that is diagonal in the computational basis, determined by $\mathcal{T}$ and $y$.
\end{lemma}
The above theorem states that the post-selected marginal distribution on $\bar{Y}$ is equal to the marginal distribution of another Gibbs distribution with a Hamiltonian $H^{(y)}$. In this paper, $d^{(i)}$ is always determined by $\mathcal{T}$ and $y$, so we suppress the dependence in notation and make the dependence implicit.
\begin{proof}
We will construct $d^{(i)}$ and show the equivalence explicitly. First, we write down the probability $P(Y,\bar{Y})$.
\begin{equation}
    P(Y=y,\bar{Y}=\bar{y}) \propto \prod_{a} e^{-\beta h_a(y,\bar{y})}
\end{equation}
We now apply the channel $\mathcal{T}$ and denote the resulting distribution as $P(\mathcal{T}[Y],\bar{Y})$, expressed as
\begin{equation}
    P(\mathcal{T}[Y]=y,\bar{Y}=\bar{y}) \propto \sum_{y'} \prod_{i \in Y} \mathcal{T}_i(y_i,y_i') \prod_{a} e^{-\beta h_a(y',\bar{y})}
\end{equation}
Where we introduce the dummy variable $y'$ that will be summed over. $y_i$ and $y_i'$ denote the value of $y$ and $y'$ on site $i$, and $\mathcal{T}_i(y_i,y_i')$ denotes the matrix element of $\mathcal{T}_i$ between $y_i$ and $y_i'$.

We will identify the dummy variable $y'$ to be the new $Y$ variable in $e^{-H^{(y)}}$ and identify $y_i$ as a parameter. We set $d^{(i)}$ so that $e^{-d^{(i)}(y_i')} = \mathcal{T}_i(y_i,y_i')$, where we use $d^{(i)}(y_i')$ to denote the value of $d^{(i)}$ under the specified $y_i'$. Since $\mathcal{T}_i(y_i,y_i')$ is non-zero, $e^{-d^{(i)}(y_i')}$ is always well-defined. Under the above notation, one can see that
\begin{equation}
    P(\mathcal{T}[Y]=y,\bar{Y}=\bar{y}) \propto \sum_{y'} \prod_{i \in Y} e^{-d^{(i)}(y_i')} \prod_{a} e^{-\beta h_a(y',\bar{y})}
\end{equation}
The left-hand side is $\mathcal{T} [\rho]\bigr|_{y}$ by definition, while the right-hand side is exactly the matrix element of $\text{Tr}_Y[e^{-H^{(y)}}]$.
\end{proof}
The above lemma states that the post-selected marginal distribution on $\bar{Y}$ is identical to the marginal distribution of a new Gibbs distribution where the Hamiltonian contains additional local ``pinning'' terms. This is a crucial step in our proof because while in general $\mathcal{T}[\rho]$ is not a Gibbs distribution of a local Hamiltonian, we are able to recover the locality through post-selection. This allows us to apply the generalized cluster expansion on $e^{-H^{(y)}}$, thereby bounding the post-selected mutual information between $X$ and $Z$.

The above Lemma can be visualized using factor graphs, a common tool in visualizing graphical models. Fig. \ref{fig:factor_graph}(a) shows the factor graph of a four-site Gibbs distribution with nearest-neighbor interactions $H_{12}$, $H_{23}$, and $H_{34}$ We partition site one as $X$, site four as $Z$, and site two and three as $Y$. To represent the marginal distribution on site one and four, we simply remove the delta tensors on site two and three, shown in Fig. \ref{fig:factor_graph}(b). Fig. \ref{fig:factor_graph}(c) shows distribution after apply channels on site two and three, that is $\mathcal{T}[\rho]$. Fig. \ref{fig:factor_graph}(d) shows the marginal distribution on site one and four after post-selecting site two and three, that is $\mathcal{T}[\rho]\bigr|_{y}$. In  Fig.\ref{fig:factor_graph}(e), we identify the channel tensor, after post-selection, with the pinning terms.  This is essentially the step where we identify $e^{-d^{(i)}(y_i')} = \mathcal{T}_i(y_i,y_i')$ in the proof. Finally, after the identification, the post-selected marginal distribution on site one and two becomes the marginal distribution of a new Gibbs distribution with pinning terms, shown in Fig.\ref{fig:factor_graph}(f).

We will denote $\rho^{(y)} \propto e^{-H^{(y)}} $. Given Lemma \ref{lem:pinned_gibbs}, we can compute $I(X:Z|\mathcal{T}[Y]=y)$ by computing the mutual formation between $X$ and $Z$ in $\rho^{(y)}$. To reuse some of the machinery we have developed for CMI, we will think of the mutual information $I_{\rho^{(y)}} (X:Z)$ as the CMI $I_{\mathcal{E}_{Y}^{Tr}[\rho^{(y)}]} (X:Z|Y)$ where $\mathcal{E}_{Y}^{Tr}$ fully depolarizes $Y$.
\begin{proposition}
    \label{prop:pinned_mi_cmi}
    the mutual information is equivalent to CMI after fully depolarizing $Y$.
    \begin{equation}
       I(X:Z|\mathcal{T}[Y]=y) =  I_{\rho^{(y)}} (X:Z) = I_{\mathcal{E}_{Y}^{Tr}[\rho^{(y)}]} (X:Z|Y)
    \end{equation}
    Where $I(X:Z|\mathcal{T}[Y]=y)$ is defined in Proposition \ref{prop:post_select}
\end{proposition}

We will upper-bound $I_{\mathcal{E}_{Y}^{Tr}[\rho^{(y)}]} (X:Z|Y)$, thereby upper-bounding $I(X:Z|\mathcal{T}[Y]=y)$. We use a similar technique from the last section to upper-bound $I_{\mathcal{E}_{Y}^{Tr}[\rho^{(y)}]} (X:Z|Y)$.

\begin{proposition}
\label{prop:pinned_mi_norm_bound}
$I_{\mathcal{E}^{Tr}[\rho^{(y)}]}(X:Z|Y)$ is upper-bounded by $\norm{H^{(y)}(X:Z|\mathcal{E}_Y^{Tr}[Y])}$.
    \begin{equation}
        I_{\mathcal{E}^{Tr}[\rho^{(y)}]}(X:Z|Y) \le \norm{H^{(y)}(X:Z|\mathcal{E}_Y^{Tr}[Y])}
    \end{equation}
\end{proposition}
We will then upper-bound $\norm{H^{(y)}(X:Z|\mathcal{E}_Y^{Tr}[Y])}$ using the generalized cluster expansion.

\subsection{Generalized Cluster Expansions Under Pinning}
We now use the generalized cluster expansion to bound $I_{\mathcal{E}_{Y}^{Tr}[\rho^{(y)}]} (X:Z|Y)$. We will use the following convention to helps us determine the appropriate normalization. Let $d_i$ be the restriction of $d^{(i)}$ to the $q$-dimensional local Hilbert space, in other words $d_i$ is a $q$-by-$q$ Hermitian matrix. Let $Z_i = \text{Tr}[ e^{-d_i}]$ and  $Z_0=\prod_{i \in Y} Z_i$. To reuse some of the previous results, we will apply partial trace channels $\mathcal{E}_{Y}^{Tr}$ to $Y$ and define $H^{(y)}(X:Z|\mathcal{E}_Y^{Tr}[Y])$ as
\begin{equation}
     H^{(y)}(X:Z|\mathcal{E}_Y^{Tr}[Y]) = \log(\mathcal{E}_Y^{Tr}[\tilde{\rho}^{(y)}_{XY}]) + \log(\mathcal{E}_Y^{Tr}[\tilde{\rho}^{(y)}_{YZ}]) - \log(\mathcal{E}_Y^{Tr}[\tilde{\rho}^{(y)}_{Y}]) - \log(\mathcal{E}_Y^{Tr}[\tilde{\rho}^{(y)}_{XYZ}])
\end{equation}
Where we use the following normalization of `density. Let $\tilde{\rho}^{(y)} = e^{-H^{(y)}} \frac{q^{|Y|}}{Z_0} $ and let $\tilde{\rho}^{(y)}_{L}$ be the reduced density matrix on $L$, under the specified normalization. We choose the above normalization because at infinite temperature $\mathcal{E}^{Tr}[\tilde{\rho}^{(y)}] \bigr|_{\beta = \infty} = I$. This allows us to expand the matrix logarithm near identity later.

We will evaluate the cluster expansion of $\tilde{\rho}^{(y)}$ and $H^{(y)}(X:Z|\mathcal{E}_Y^{Tr}[Y])$ by writing down the multivariate Taylor expansions in $\lambda_a$.
\begin{align}
\mathcal{E}_Y^{Tr}[\tilde{\rho}^{(y)}_L] &= \sum_{\mathbf{W}} \frac{\lambda_{\mathbf{W}}}{\mathbf{W}!}\mathcal{D}_\mathbf{W} \mathcal{E}_Y^{Tr}[\tilde{\rho}^{(y)}_L]\\
    \log(\mathcal{E}_Y^{Tr}[\tilde{\rho}^{(y)}_L]) &= \sum_{\mathbf{W}} \frac{\lambda_{\mathbf{W}}}{\mathbf{W}!}\mathcal{D}_\mathbf{W} \log(\mathcal{E}_Y^{Tr}[\tilde{\rho}^{(y)}_L])
\end{align}
Crucially, the pinning terms $d^{(i)}$ are not variables to be expanded! In other words, the zeroth order term in the cluster expansion of $\mathcal{E}^{Tr}[\tilde{\rho}^{(y)}_L]$ is $\frac{q^{|Y|}}{Z_0} \mathcal{E}^{Tr}[e^{\sum_{i \in Y} d^{(i)}}]$. In fact, $d^{(i)}$ cannot participate in the cluster expansion because they have large operator norms in general, equivalently the pinning terms can be at low temperatures.

One may worry that the presence of the low-temperature pinning terms could destroy the Taylor expansion in the matrix logarithm, but since all sites that are pinned are also traced out by $\mathcal{E}_Y^{Tr}$, the zeroth order term $\frac{q^{|Y|}}{Z_0} \mathcal{E}^{Tr}[e^{\sum_{i \in Y} d^{(i)}}]$ becomes identity, so we are still Taylor expanding the matrix logarithm around the identity.

\subsection{Connected Clusters}
We will derive simplifications in connected clusters, similar to Lemma \ref{lem:conn_cluster} and Lemma \ref{lem:conn_log_cluster}. Crucially, since the pinning terms are all local, their presence do not interfere with the desired property of connected clusters.

\begin{lemma}
    \label{lem:pinned_conn_cluster}
 Let $\mathcal{E}_{\Gamma}^{Tr} = \otimes_{i \in \Gamma} \mathcal{E}^{Tr}_{i}$, where $\Gamma$ denotes a larger region that contains $Y$. $\mathcal{E}^{Tr}_{i}$ denotes the partial trace channel acting on site $i$ in $\Gamma$. Suppose $\mathbf{W}$ can be decomposed into two disconnected clusters $\mathbf{W}_1$ and $\mathbf{W}_2$. Then,
\begin{align}
    \mathcal{D}_\mathbf{W} \mathcal{E}_{\Gamma}^{Tr}[\tilde{\rho}^{(y)}]= \mathcal{D}_{\mathbf{W}_1}  \mathcal{E}_{\Gamma}^{Tr}[\tilde{\rho}^{(y)}] \times \mathcal{D}_\mathbf{W_2} \mathcal{E}_{\Gamma}^{Tr}[\tilde{\rho}^{(y)}]
\end{align}
\end{lemma}
\begin{proof}
Since $\mathbf{W}$ can be decomposed into two disconnected clusters $\mathbf{W}_1$ and $\mathbf{W}_2$. We have
    \begin{align}
        \mathcal{D}_\mathbf{W} \mathcal{E}_{\Gamma}^{Tr}[\tilde{\rho}^{(y)}] = \prod_{i \in  \mathbf{W}_1} (\frac{\partial }{\partial \lambda_i})^{\mu_i}  \prod_{j \in  \mathbf{W}_2} (\frac{\partial }{\partial \lambda_j})^{\mu_j}  \mathcal{E}_{\Gamma}^{Tr}[\tilde{\rho}^{(y)}]\Bigr|_{\lambda=0}
    \end{align}
 We first set the coefficients $\lambda$ that are not in $\mathbf{W}_1$ or $\mathbf{W}_2$ to zero. Because $\mathbf{W}_1$ and $\mathbf{W}_2$ are disconnected, the resulting density matrix factorizes.
      \begin{align}
        \mathcal{D}_\mathbf{W} \mathcal{E}_{\Gamma}^{Tr}[\tilde{\rho}^{(y)}] &= \prod_{i \in  \mathbf{W}_1} (\frac{\partial }{\partial \lambda_i})^{\mu_i}  \prod_{j \in  \mathbf{W}_2} (\frac{\partial }{\partial \lambda_j})^{\mu_j}\mathcal{E}_{\Gamma}^{Tr}\left[\tilde{\rho}^{(\mathbf{W}_1)}\times \tilde{\rho}^{(\mathbf{W}_2)}\times \tilde{\rho}^{(\bar{\mathbf{W}})}\right] \Bigr|_{\lambda_i=\lambda_j=0} \\
        \tilde{\rho}^{(\mathbf{W}_1)} &= e^{-\beta\sum_{a \in \mathbf{W}_1} \lambda_a h_a} \prod_{i \in \Gamma \cap \text{supp}(\mathbf{W}_1)} \left( \frac{ q}{Z_i} e^{- d^{(i)}} \right) \label{eq:w1_def}  \\
        \tilde{\rho}^{(\mathbf{W}_2)} &= e^{-\beta\sum_{a \in \mathbf{W}_2} \lambda_a h_a} \prod_{i \in \Gamma \cap \text{supp}(\mathbf{W}_1)} \left( \frac{ q}{Z_i} e^{- d^{(i)}} \right) \label{eq:w2_def} \\
        \tilde{\rho}^{(\bar{\mathbf{W}})} &=  \prod_{i \in \Gamma \cap \text{supp}(\bar{\mathbf{W}})} \left( \frac{ q}{Z_i} e^{- d^{(i)}} \right) \label{eq:barw_def}
    \end{align}
Where we use $\text{supp}(\bar{\mathbf{W}})$ to denote the complement to the support of $\mathbf{W}_1 \cup \mathbf{W}_2$. In short, we separate into three parts. The first part $\tilde{\rho}^{(\mathbf{W}_1)}$ contains all $h_a$ and $d_i$ supported on $\mathbf{W}_1$, the second part $\tilde{\rho}^{(\mathbf{W}_2)}$ contains all $h_a$ and $d_i$ supported on $\mathbf{W}_2$, and the last part $\tilde{\rho}^{(\bar{\mathbf{W}})}$ contains all $d_i$ not supported on $\mathbf{W}_1 \cup \mathbf{W}_2$. We use the following convention of $d^{(i)}$ to simplify the notation: $d^{(i)}=0$ and $Z_i = q$ if $i \in \Gamma$ but $i \notin Y$.

We decompose $\mathcal{E}_{\Gamma}^{Tr} = \mathcal{E}_{\mathbf{W}_1} \otimes \mathcal{E}_{\mathbf{W}_2} \otimes \mathcal{E}_{\bar{\mathbf{W}}}$, where $\mathcal{E}_{\mathbf{W}_1}$ includes partial trace channels acting on the intersection between $\Gamma$ and the support of $\mathbf{W}_1$, that is $\Gamma \cap \text{supp}(\mathbf{W}_1)$, $\mathcal{E}_{\mathbf{W}_2}$ includes partial trace channels acting on $\Gamma \cap \text{supp}(\mathbf{W}_2)$, and $\mathcal{E}_{\bar{\mathbf{W}}}$ includes partial trace channels acting on the remaining system. Let $\mathcal{E}^{({\mathbf{W}_1})} = \mathcal{E}_{\mathbf{W}_1} \otimes I_{\mathbf{W_2 \cup \bar{W}}}$ be the embedding of $\mathcal{E}_{\mathbf{W}_1}$ into the global Hilbert space by tensoring with the identity channel on the remaining system. $\mathcal{E}^{(\mathbf{W_2})}$ and $\mathcal{E}^{(\bar{\mathbf{W}})}$ are similarly defined. Under the above factorization, we have
    \begin{equation}
\mathcal{D}_\mathbf{W} \mathcal{E}_{\Gamma}^{Tr}\left[\tilde{\rho}^{(\mathbf{W}_1)}\times \tilde{\rho}^{(\mathbf{W}_2)}\times \tilde{\rho}^{(\bar{\mathbf{W}})}\right]  = \mathcal{D}_{\mathbf{W}}\left( \mathcal{E}^{(\mathbf{W}_1)} \left[\tilde{\rho}^{(\mathbf{W}_1)} \right] \times \mathcal{E}^{(\mathbf{W}_2)} \left[\tilde{\rho}^{(\mathbf{W}_2)} \right] \times \mathcal{E}^{(\bar{\mathbf{W}})} \left[\tilde{\rho}^{(\bar{\mathbf{W}})} \right] \right)
    \end{equation}
And one can see that only $\tilde{\rho}^{(\mathbf{W}_1)}$ depends on $\mathbf{W}_1$ and only $\tilde{\rho}^{(\mathbf{W}_2)}$ depends on $\mathbf{W}_2$, so the cluster derivative factorizes into
    \begin{equation} \label{eq:pinned_cluster_factorize}
\mathcal{D}_\mathbf{W} \mathcal{E}_{\Gamma}^{Tr}\left[\tilde{\rho}^{(\mathbf{W}_1)}\times \tilde{\rho}^{(\mathbf{W}_2)}\times \tilde{\rho}^{(\bar{\mathbf{W}})}\right]  = \mathcal{D}_{\mathbf{W}_1} \mathcal{E}^{(\mathbf{W}_1)} \left[\tilde{\rho}^{(\mathbf{W}_1)} \right] \times \mathcal{D}_{\mathbf{W}_2} \mathcal{E}^{(\mathbf{W}_2)} \left[\tilde{\rho}^{(\mathbf{W}_2)} \right] \times \mathcal{E}^{(\bar{\mathbf{W}})} \left[\tilde{\rho}^{(\bar{\mathbf{W}})} \right] 
    \end{equation}
Lastly, we show that the above equation is indeed a product of cluster derivatives $\mathcal{D}_{\mathbf{W}_1}  \mathcal{E}^{({\mathbf{W}_1})}[\tilde{\rho}^{(y)}] \times \mathcal{D}_\mathbf{W_2} \mathcal{E}^{(\mathbf{W_2})}[\tilde{\rho}^{(y)}]$ with the correct normalization. First, One can verify that $\mathcal{E}^{(\bar{\mathbf{W}})} \left[\tilde{\rho}^{(\bar{\mathbf{W}})} \right] = I$. This is shown by explicitly evaluating the partial trace channels
\begin{equation}\label{eq:barw_identity}
    \mathcal{E}^{(\bar{\mathbf{W}})} \left[\tilde{\rho}^{(\bar{\mathbf{W}})} \right] = \prod_{i \in \Gamma \cap \text{supp}(\bar{\mathbf{W}})} \left( \frac{ q}{Z_i} \mathcal{E}^{Tr}_{i}\left[e^{- d^{(i)}}\right] \right) = \prod_{i \in \Gamma \cap \text{supp}(\bar{\mathbf{W}})} \left( \frac{ q}{Z_i}  \frac{Z_i}{q} I \right) = I
\end{equation}
Next, we show that $\mathcal{D}_{\mathbf{W}_1}  \mathcal{E}_{\Gamma}^{Tr}[\tilde{\rho}^{(y)}] = \mathcal{E}^{(\mathbf{W}_1)} \left[\tilde{\rho}^{(\mathbf{W}_1)} \right]$. We write down $\mathcal{D}_{\mathbf{W}_1}  \mathcal{E}^{Tr}[\tilde{\rho}^{(y)}]$ explicitly.
\begin{align}
    \mathcal{D}_{\mathbf{W}_1} \mathcal{E}_{\Gamma}^{Tr}[\tilde{\rho}^{(y)}] &= \prod_{i \in  \mathbf{W}_1} (\frac{\partial }{\partial \lambda_i})^{\mu_i}  \mathcal{E}_{\Gamma}^{Tr}\left[\tilde{\rho}^{(\mathbf{W}_1)} \times \tilde{\rho}^{(\bar{\mathbf{W}}_1)}\right] \Bigr|_{\lambda_i=\lambda_j=0} \\
        \tilde{\rho}^{(\mathbf{W}_1)} &= e^{-\beta\sum_{a \in \mathbf{W}_1} \lambda_a h_a} \prod_{i \in \Gamma \cap \text{supp}(\mathbf{W}_1)} \left( \frac{ q}{Z_i} e^{- d^{(i)}} \right) \\
        \tilde{\rho}^{(\bar{\mathbf{W}}_1)} &=  \prod_{i \in \Gamma \cap \text{supp}(\bar{\mathbf{W}}_1)} \left( \frac{ q}{Z_i} e^{- d^{(i)}} \right)
\end{align}
Where we use $\text{supp}(\bar{\mathbf{W}}_1)$ to denote the complement to the support of $\mathbf{W}_1$. In short, we separate into two parts. The first part $\tilde{\rho}^{(\mathbf{W}_1)}$ contains all $h_a$ and $d_i$ supported on $\mathbf{W}_1$ and is identical to the $\tilde{\rho}^{(\mathbf{W}_1)}$ defined in Eq.(\ref{eq:w1_def}). The second part $\tilde{\rho}^{(\bar{\mathbf{W}}_1)}$ contains all $d_i$ not supported on $\mathbf{W}_1$. 

We decompose $\mathcal{E}_{\Gamma}^{Tr} = \mathcal{E}_{\mathbf{W}_1} \otimes  \mathcal{E}_{\bar{\mathbf{W}}_1}$, where $\mathcal{E}_{\mathbf{W}_1}$ includes partial trace channels acting on $\Gamma \cap \text{supp}(\mathbf{W}_1)$ and is the same as previously defined. $\mathcal{E}_{\bar{\mathbf{W}}}$ includes partial trace channels acting on the remaining system. Let $\mathcal{E}^{({\mathbf{W}_1})} = \mathcal{E}_{\mathbf{W}_1} \otimes I_{\mathbf{W_2 \cup \bar{W}}}$ be the embedding of $\mathcal{E}_{\mathbf{W}_1}$ into the global Hilbert space by tensoring with the identity channel on the remaining system. $\mathcal{E}^{(\mathbf{W_2})}$ and $\mathcal{E}^{(\bar{\mathbf{W}})}$ are similarly defined. Under the above factorization, we have
    \begin{equation}
\mathcal{E}_{\Gamma}^{Tr}\left[\tilde{\rho}^{(\mathbf{W}_1)}\times \tilde{\rho}^{(\bar{\mathbf{W}}_1)}\right]  = \mathcal{E}^{(\mathbf{W}_1)} \left[\tilde{\rho}^{(\mathbf{W}_1)} \right]  \times \mathcal{E}^{(\bar{\mathbf{W}}_1)} \left[\tilde{\rho}^{(\bar{\mathbf{W}}_1)} \right] 
    \end{equation}
Following the same math in Eq. (\ref{eq:barw_identity}), we can verify that $\mathcal{E}^{(\bar{\mathbf{W}}_1)} \left[\tilde{\rho}^{(\bar{\mathbf{W}}_1)} \right] = I$. Therefore, we have $\mathcal{D}_{\mathbf{W}_1}  \mathcal{E}_{\Gamma}^{Tr}[\tilde{\rho}^{(y)}] = \mathcal{E}^{(\mathbf{W}_1)} \left[\tilde{\rho}^{(\mathbf{W}_1)} \right]$. Similarly, one can show that $\mathcal{D}_{\mathbf{W}_2}  \mathcal{E}_{\Gamma}^{Tr}[\tilde{\rho}^{(y)}] = \mathcal{E}^{(\mathbf{W}_2)} \left[\tilde{\rho}^{(\mathbf{W}_2)} \right]$. Plugging the two equivalence and $\mathcal{E}^{(\bar{\mathbf{W}})} \left[\tilde{\rho}^{(\bar{\mathbf{W}})} \right] = I$ into Eq. (\ref{eq:pinned_cluster_factorize}), we arrive at the final result.
\end{proof}
Note that Lemma \ref{lem:pinned_conn_cluster} and \ref{lem:conn_cluster} are slightly different as in Lemma \ref{lem:pinned_conn_cluster} the complete partial trace channel shows up in the right-hand side, whereas only factorized channels show up in in the right-hand side of Lemma \ref{lem:conn_cluster}. Lemma \ref{lem:pinned_conn_cluster} will also not work for channels that do not trace out the entire $Y$.

\begin{lemma}
\label{lem:pinned_conn_log_cluster}
Let $\mathcal{E} = \mathcal{E}_1 \otimes \mathcal{E}_2 \otimes \ldots \otimes \mathcal{E}_n$ be a product of local channels. When $\mathbf{W}$ is disconnected, $\mathcal{D}_\mathbf{W} \log(\mathcal{E}[\tilde{\rho}^{(y)}])=0$
\end{lemma}
\begin{proof}
Suppose $\mathbf{W}$ can be decomposed into two disconnected clusters $\mathbf{W}_1$ and $\mathbf{W}_2$. We have
\begin{align}
    \mathcal{D}_\mathbf{W} \log \left(\mathcal{E}[\tilde{\rho}^{(y)}]\right) = \prod_{i \in  \mathbf{W}_1} (\frac{\partial }{\partial \lambda_i})^{\mu_i}  \prod_{j \in  \mathbf{W}_2} (\frac{\partial }{\partial \lambda_j})^{\mu_j}  \log\left(\mathcal{E}[\tilde{\rho}^{(y)}]\right) \Bigr|_{\lambda=0}
\end{align}
We first set the coefficients $\lambda$ that are not in $\mathbf{W}_1$ or $\mathbf{W}_2$ to zero. Because $\mathbf{W}_1$ and $\mathbf{W}_2$ are disconnected, the resulting density matrix factorizes.
  \begin{equation}
    \mathcal{D}_\mathbf{W} \log \left(\mathcal{E}[\tilde{\rho}^{(y)}]\right) = \prod_{i \in  \mathbf{W}_1} (\frac{\partial }{\partial \lambda_i})^{\mu_i}  \prod_{j \in  \mathbf{W}_2} (\frac{\partial }{\partial \lambda_j})^{\mu_j}\log\left(\mathcal{E}\left[\tilde{\rho}^{(\mathbf{W}_1)}\times \tilde{\rho}^{(\mathbf{W}_2)}\times \tilde{\rho}^{(\bar{\mathbf{W}})}\right]\right) \Bigr|_{\lambda_i=\lambda_j=0}
\end{equation}

Where $\tilde{\rho}^{(\mathbf{W}_1)}$, $\tilde{\rho}^{(\mathbf{W}_2)}$, and $\tilde{\rho}^{(\bar{\mathbf{W}})}$ are defined in Eq. (\ref{eq:w1_def},\ref{eq:w2_def},\ref{eq:barw_def}), respectively. We factorize $\mathcal{E}^{Tr}$ into $\mathcal{E}^{Tr} = \mathcal{E}_{\mathbf{W}_1} \otimes \mathcal{E}_{\mathbf{W}_2} \otimes \mathcal{E}_{\bar{\mathbf{W}}}$ as done in Lemma \ref{lem:conn_log_cluster}. We also define $\mathcal{E}^{({\mathbf{W}_1})}$, $\mathcal{E}^{(\mathbf{W_2})}$ and $\mathcal{E}^{(\bar{\mathbf{W}})}$ accordingly. Under the factorization, the logarithm becomes
\begin{align}
    \log\left(\mathcal{E}\left[\tilde{\rho}^{(\mathbf{W}_1)}\times \tilde{\rho}^{(\mathbf{W}_2)}\times \tilde{\rho}^{(\bar{\mathbf{W}})}\right]\right) &= \log\left(\mathcal{E}^{(\mathbf{W}_1)}\left[\tilde{\rho}^{(\mathbf{W}_1)}\right]\times \mathcal{E}^{(\mathbf{W}_2)}\left[\tilde{\rho}^{(\mathbf{W}_2)}\right] \times \mathcal{E}^{(\bar{\mathbf{W}})}\left[\tilde{\rho}^{(\bar{\mathbf{W}})}\right]\right) \\
    &= \log\left(\mathcal{E}^{(\mathbf{W}_1)}\left[\tilde{\rho}^{(\mathbf{W}_1)}\right] \right) + \log \left( \mathcal{E}^{(\mathbf{W}_2)}\left[\tilde{\rho}^{(\mathbf{W}_2)}\right] \right) + \log \left( \mathcal{E}^{(\bar{\mathbf{W}})}\left[\tilde{\rho}^{(\bar{\mathbf{W}})}\right]\right)
\end{align}
After expressing the matrix logarithm as a sum of three terms, notice that the first terms only depend on $\mathbf{W}_1$, the second term only depends $\mathbf{W}_2$, and the last term is independent of $\mathbf{W}_1$ or $\mathbf{W}_2$. Therefore, all three terms become zero after taking the cluster derivative of $\mathbf{W} = \mathbf{W}_1 \cup \mathbf{W}_2$.
\end{proof}
When considering the cluster derivative of $H^{(y)}(X:Z|\mathcal{E}_Y^{Tr}[Y])$, we can show that only connected clusters connecting $X$ and $Z$ contribute.
\begin{lemma}\label{lem:pinned_conn_AC}
Let $\mathcal{E} = \mathcal{E}_1 \otimes \mathcal{E}_2 \otimes \ldots \otimes \mathcal{E}_n$ be a product of local channels. $\mathcal{D}_\mathbf{W} H^{(y)}(X:Z|\mathcal{E}_Y^{Tr}[Y])$ is non-trivial only when $\mathbf{W}$ is a connected cluster connecting $X$ and $Z$.
\end{lemma}
\begin{proof}
    If $\mathbf{W}$ is disconnected, $\mathcal{D}_\mathbf{W} H^{(y)}(X:Z|\mathcal{E}_Y^{Tr}[Y]) = 0$ follows from Lemma \ref{lem:pinned_conn_log_cluster}. Without loss of generality, suppose that $\mathbf{W}$ is connected but does not connect to $Z$, then we first set all coefficients of terms supported on $Z$ to zero. We first consider $\mathcal{D}_\mathbf{W}(\log(\mathcal{E}[\tilde{\rho}_{XY}^{(y)}])$ and $\mathcal{D}_\mathbf{W}(\log(\mathcal{E}[\tilde{\rho}_{XYZ}^{(y)}])$. $\mathcal{D}_\mathbf{W}(\log(\mathcal{E}[\tilde{\rho}_{XY}^{(y)}])$ can be written as
\begin{align}
        \mathcal{D}_\mathbf{W} \log(\mathcal{E}[\tilde{\rho}_{XY}^{(y)}]) = \mathcal{D}_\mathbf{W} \log((\mathcal{E} \circ \mathcal{E}^{Tr}_{C}) [\tilde{\rho}_{XYZ}^{(y)}]) \\
        = \mathcal{D}_\mathbf{W} \log((\mathcal{E} \circ \mathcal{E}^{Tr}_{C}) [\tilde{\rho}^{(W)}])
\end{align}
Where $\tilde{\rho}^{(W)}$ is defined in Eq. (\ref{eq:w1_def}) but replacing $\mathbf{W}_1$ with $\mathbf{W}$. $\mathcal{E}^{Tr}_{Z}$ denotes the partial trace channels acting on $Z$ and $\circ$ denotes channel composition. Next, we recognize that $\tilde{\rho}^{(W)}$ is identity on $Z$, so the action of $\mathcal{E}^{Tr}_{Z}$ does not change $\tilde{\rho}^{(W)}$ at all. Therefore,
\begin{equation}
        \mathcal{D}_\mathbf{W} \log(\mathcal{E}[\tilde{\rho}_{XY}^{(y)}]) = \mathcal{D}_\mathbf{W} \log(\mathcal{E}  [\tilde{\rho}^{(W)}])
\end{equation}
Similarly, $\mathcal{D}_\mathbf{W}(\log(\mathcal{E}[\tilde{\rho}_{XYZ}^{(y)}])$ can be written as
\begin{align}
        \mathcal{D}_\mathbf{W} \log(\mathcal{E}[\tilde{\rho}_{XYZ}]) = \mathcal{D}_\mathbf{W} \log(\mathcal{E} [\tilde{\rho}^{(W)}])
\end{align}
Therefore, $\mathcal{D}_\mathbf{W}\left(\log(\mathcal{E}[\tilde{\rho}_{XY}^{(y)}]) - \log(\mathcal{E}[\tilde{\rho}_{XYZ}^{(y)}]) \right) = 0$. Next, we consider $\mathcal{D}_\mathbf{W}\log(\mathcal{E}[\tilde{\rho}_{YZ}^{(y)}])$ and $\mathcal{D}_\mathbf{W}\log(\mathcal{E}[\tilde{\rho}_{Y}^{(y)}])$. $\mathcal{D}_\mathbf{W}\log(\mathcal{E}[\tilde{\rho}_{YZ}^{(y)}])$ can be written as
\begin{align}
      \mathcal{D}_\mathbf{W}(\log(\mathcal{E}[\tilde{\rho}_{YZ}^{(y)}]) = \mathcal{D}_\mathbf{W} \log((\mathcal{E} \circ \mathcal{E}^{Tr}_{X}) [\tilde{\rho}_{YZ}^{(y)}]) \\
        = \mathcal{D}_\mathbf{W} \log((\mathcal{E} \circ \mathcal{E}^{Tr}_{X}) [\tilde{\rho}^{(W)}])
\end{align}
Similarly, $\mathcal{D}_\mathbf{W}\log(\mathcal{E}[\tilde{\rho}_{Y}^{(y)}])$ can be written as
\begin{align}
      \mathcal{D}_\mathbf{W}(\log(\mathcal{E}[\tilde{\rho}_{Y}^{(y)}]) = \mathcal{D}_\mathbf{W} \log((\mathcal{E} \circ \mathcal{E}^{Tr}_{X}\circ \mathcal{E}^{Tr}_{Z}) [\tilde{\rho}_{XYZ}^{(y)}]) \\
        = \mathcal{D}_\mathbf{W} \log((\mathcal{E} \circ \mathcal{E}^{Tr}_{X}\circ \mathcal{E}^{Tr}_{Z}) [\tilde{\rho}^{(W)}]) = \mathcal{D}_\mathbf{W} \log((\mathcal{E} \circ \mathcal{E}^{Tr}_{X}) [\tilde{\rho}^{(W)}])
\end{align}
Where in the last line we use again the fact that $\tilde{\rho}^{(W)}$ is identity on $Z$ and the action of $\mathcal{E}^{Tr}_{Z}$ does not change $\tilde{\rho}^{(W)}$. Therefore, $\mathcal{D}_\mathbf{W}\left(\log(\mathcal{E}[\tilde{\rho}_{YZ}^{(y)}]) - \log(\mathcal{E}[\tilde{\rho}_{Y}^{(y)}]) \right) = 0$. Together, we establish that $\mathcal{D}_\mathbf{W} H^{(y)}(X:Z|\mathcal{E}_Y^{Tr}[Y]) = 0$.
\end{proof}
We have seen that only connected clusters contribute to the generalized cluster expansion, even in the presence of local pinning terms. Therefore, we can apply Lemma \ref{lem:num_conn_cluster} to show that only exponentially many terms contribute to the cluster expansion. By showing the exponential decay of each term in the next subsection, we establish the convergence of the generalized cluster expansion.

\subsection{Bounding Cluster Derivative Under Pinning}
In this subsection, we upper-bound the operator norm of each term in the cluster expansion. The technique is mostly identical to the quantum case. Note that classical distributions automatically commute, so we do not have to worry about the ordering issue in connecting to the graph coloring.
\begin{lemma}
\label{lem:pinned_cluster_deriv_upper}
Let $\mathcal{E}_\Gamma^{Tr} = \otimes_{i \in \Gamma} \mathcal{E}^{Tr}_{i}$, where $\Gamma$ denotes a larger region that contains $Y$ and $\mathcal{E}^{Tr}_{i}$ denotes the partial trace channel acting on site $i$ in $\Gamma$.  Given any cluster $\mathbf{W}$ with order $|\mathbf{W}|=m$,  we have
\begin{align}
    \frac{1}{\mathbf{W}!}\norm{\mathcal{D}_{\mathbf{W}} \log(\mathcal{E}_\Gamma^{Tr}[\tilde{\rho}^{(y)}])} \le (2e(\mathfrak{d}+1)\beta)^{m+1}
\end{align}
\end{lemma}

To prove Lemma \ref{lem:pinned_cluster_deriv_upper}, we first supply a short lemma that upper-bounds $\norm{\mathcal{D}_\mathbf{V}\mathcal{E}_\Gamma^{Tr}[\tilde{\rho}^{(y)}]}$, similar to Lemma \ref{lem:deriv_bound}.
\begin{lemma}\label{lem:pinned_deriv_bound}
Let $\mathcal{E}_\Gamma^{Tr} = \otimes_{i \in \Gamma} \mathcal{E}^{Tr}_{i}$, where $\Gamma$ denotes a larger region that contains $Y$. $\mathcal{E}^{Tr}_{i}$ denotes the partial trace channel acting on site $i$ in $\Gamma$. Suppose $|\mathbf{V}|=k$, then $\norm{\mathcal{D}_\mathbf{V}\mathcal{E}_\Gamma^{Tr}[\tilde{\rho}^{(y)}]} \le \beta^k$
\end{lemma}
\begin{proof}
    First, expand $\tilde{\rho}^{(y)} = \sum_{n=1}^{\infty}\frac{(-\beta)^n}{n!} \left( H^n \prod_{i \in \Gamma} \frac{q}{Z_i} e^{- d^{(i)}} \right)$, where again we use the convention that $d^{(i)}=0$ and $Z_i = q$ if $i \in \Gamma$ but $i \notin Y$.
\begin{align}
    \norm{\mathcal{D}_\mathbf{V}\mathcal{E}_\Gamma^{Tr}[\tilde{\rho}]} = \norm{\mathcal{D}_\mathbf{V}\mathcal{E}_\Gamma^{Tr}\left[\sum_{n=1}^{\infty}\frac{(-\beta)^n}{n!} \left( H^n \prod_{i \in Y} \frac{q}{Z_i}  e^{- d^{(i)}} \right)\right]} \le \sum_{n=1}^{\infty} \frac{\beta^n}{n!} \norm{\mathcal{D}_\mathbf{V}\mathcal{E}_\Gamma^{Tr}\left[H^n \prod_{i \in Y} \frac{q}{Z_i}  e^{- d^{(i)}}\right]}
\end{align}
Suppose $\mathbf{V}$ contain terms $h^{(\mathbf{V})}_1, h^{(\mathbf{V})}_2, \ldots, h^{(\mathbf{V})}_k$, where one term can show up multiple times to account for its multiplicity. Only the term on the $k$-th order contributes, so we have
\begin{align}
    \norm{\mathcal{D}_\mathbf{V}\mathcal{E}_\Gamma^{Tr}[\tilde{\rho}]} \le \frac{\beta^k}{k!} \norm{  \mathcal{D}_\mathbf{V}\mathcal{E}_\Gamma^{Tr}\left[H^k \prod_{i \in Y} \frac{q}{Z_i}  e^{- d^{(i)}}\right]} = \beta^k \norm{ \mathcal{E}_\Gamma^{Tr}\left[h^{(V)}_{1}h^{(V)}_{2} \ldots h^{(V)}_{k}\prod_{i \in Y} \frac{q}{Z_i}  e^{- d^{(i)}}\right]}
\end{align}
Where the ordering of $h^{(\mathbf{V})}_1, h^{(\mathbf{V})}_2, \ldots, h^{(\mathbf{V})}_k$ does not matter because they all commute. To proceed, we factorize ${h^{(V)}_{1}h^{(V)}_{2} \ldots h^{(V)}_{k}}$, which is a product of local operators in $\{D_a\}$ with operator norms bounded by one, into the following form
\begin{equation}
    h^{(V)}_{1}h^{(V)}_{2} \ldots h^{(V)}_{k} = h_{\Gamma} h_{\bar{\Gamma}}
\end{equation}
Where $h_{\Gamma}$ is supported on $\Gamma$ and $h_{\bar{\Gamma}}$ is supported on the complement $\bar{\Gamma}$. In this way, the previous bound becomes
\begin{equation}
    \beta^k \norm{ \mathcal{E}_\Gamma^{Tr}\left[h^{(V)}_{1}h^{(V)}_{2} \ldots h^{(V)}_{k}\prod_{i \in \Gamma} \frac{q}{Z_i}  e^{- d^{(i)}}\right]} = \beta^k \norm{h_{\bar{\Gamma}}} \cdot \norm{\mathcal{E}_\Gamma^{Tr}\left[ h_{\Gamma} \prod_{i \in \Gamma} \frac{q}{Z_i}  e^{- d^{(i)}} \right]}
\end{equation}
Since $h_{\Gamma}$ is a product of operators with operator norm bounded by one, by the sub-multiplicity of operator norm, $\norm{h_{\Gamma}} \le 1$. On the other hand, the second operator can be directly evaluated.
\begin{equation}
    \mathcal{E}_\Gamma^{Tr}\left[ h_{\Gamma} \prod_{i \in \Gamma} \frac{q}{Z_i}  e^{- d^{(i)}} \right] = \frac{1}{q^{|\Gamma |}} \text{Tr}_{\Gamma}\left[ h_{\Gamma} \prod_{i \in \Gamma} \frac{q}{Z_i}  e^{- d^{(i)}} \right] \otimes I_{\Gamma}
\end{equation}
Where $|\Gamma|$ denotes the size of $\Gamma$, and $I_{\Gamma}$ denotes the $q^{|\Gamma |}$-by-$q^{|\Gamma |}$ dimensional identity operator on the local Hilbert space of $\Gamma$. One can quickly see that $|\text{Tr}_{\Gamma}\left[ h_{\Gamma} \prod_{i \in \Gamma} \frac{q}{Z_i}  e^{- d^{(i)}} \right]| \le q^{|\Gamma|}$ because $\frac{q}{Z_i}  e^{- d^{(i)}}$ is a positive diagonal matrix in the computational basis with a trace of $q^{|\Gamma|}$. On the other hand, $h_{\Gamma}$ is also diagonal in the computational basis with diagonal elements in the range of $[-1, +1]$. Therefore, $h_{\Gamma} \prod_{i \in \Gamma} \frac{q}{Z_i}  e^{- d^{(i)}}$ cannot have an absolute value of trace larger than $q^{|\Gamma|}$. Thus, $\norm{\mathcal{E}_\Gamma^{Tr}\left[ h_{\Gamma} \prod_{i \in \Gamma} \frac{q}{Z_i}  e^{- d^{(i)}} \right]} \le 1$, so we have $\norm{\mathcal{D}_\mathbf{V}\mathcal{E}_\Gamma^{Tr}[\tilde{\rho}^{(y)}]} \le \beta^k$ in the end.
\end{proof}

Similar to the quantum case, we first connect the cluster derivative to the graph coloring problem.
\begin{lemma}
\label{lem:pinned_graph_coloring}
Let $\mathcal{E}_{\Gamma}^{Tr} = \otimes_{i \in \Gamma} \mathcal{E}^{Tr}_{i}$, where $\Gamma$ denotes a larger region that contains $Y$ and $\mathcal{E}^{Tr}_{i}$ denotes the partial trace channel acting on site $i$ in $\Gamma$. Given any cluster $\mathbf{W}$ with order $|\mathcal{W}|=k$, we have
\begin{align}
        \norm{\mathcal{D}_{\mathbf{W}} \log(\mathcal{E}_{\Gamma}^{Tr}[\tilde{\rho}^{(y)}])} \le   \beta^{|\mathbf{W}|} \left( \sum_{B  \in \text{PaC}(\text{Gra}(\mathbf{W}))}   \sum_{n=1}^{|B|} \frac{(-1)^n}{n}\chi^*(n,\text{Gra}(B)) \right)
    \end{align}
\end{lemma}
\begin{proof}
  The proof is largely similar to the proof of Lemma \ref{lem:graph_coloring}.  We first cluster expand $\log(\mathcal{E}_{\Gamma}^{Tr}[\tilde{\rho}^{(y)}])$ 
\begin{align}
     \mathcal{E}_{\Gamma}^{Tr}[\tilde{\rho}^{(y)}] =\prod_{i \in \Gamma} \mathcal{E}_{i}^{Tr}\left[\frac{q}{Z_i} e^{- d^{(i)}}\right]+ \sum_{k=1}^{\infty}  \sum_{\mathbf{W}:|\mathbf{W}| =k} \frac{\lambda_{\mathbf{W}}}{\mathbf{W}!}\mathcal{D}_\mathbf{W}\mathcal{E}_{\Gamma}^{Tr}[\tilde{\rho}^{(y)}] = I+ \sum_{k=1}^{\infty}  \sum_{\mathbf{W}:|\mathbf{W}| =k} \frac{\lambda_{\mathbf{W}}}{\mathbf{W}!}\mathcal{D}_\mathbf{W}\mathcal{E}_{\Gamma}^{Tr}[\tilde{\rho}^{(y)}]
\end{align}
Where we use the convention that $d^{(i)}=0$ and $Z_i = q$ if $i \in \Gamma$ but $i \notin Y$. The second equality follows from Eq. (\ref{eq:barw_identity}).

In general, $\mathbf{W}$ may be disconnected. We use $P_{max}(\mathbf{W})$ to denote the maximally connected subset of $\mathbf{W}$, namely the minimal partition that separates $\mathbf{W}$ into connected subsets. Using Lemma \ref{lem:pinned_conn_cluster}, we have
\begin{align}
   \mathcal{E}_{\Gamma}^{Tr}[\tilde{\rho}^{(y)}] =  I+ \sum_{k=1}^{\infty}  \sum_{\mathbf{W}:|\mathbf{W}| =k} \prod_{\mathbf{V} \in P_{max}(\mathbf{W})} (\frac{\lambda_{\mathbf{V}}}{\mathbf{V}!}\mathcal{D}_\mathbf{V}\mathcal{E}_{\Gamma}^{Tr}[\tilde{\rho}^{(y)}])
\end{align}
Next, we apply the matrix logarithm expansion $\log(I+A)=\sum_{n=1}^{\infty} \frac{(-1)^{n-1}}{n}A^n$ to expand $\log(\mathcal{E}_{\Gamma}^{Tr}[\tilde{\rho}^{(y)}])$.
\begin{align}
    \log(\mathcal{E}_{\Gamma}^{Tr}[\tilde{\rho}^{(y)}]) &= \sum_{n=1}^{\infty} \frac{(-1)^{n-1}}{n} \left(\sum_{k=1}^{\infty}   \sum_{\mathbf{W}:|\mathbf{W}| =k} \frac{\lambda_{\mathbf{W}}}{\mathbf{W}!}\mathcal{D}_\mathbf{W}\mathcal{E}_{\Gamma}^{Tr}[\tilde{\rho}^{(y)}]\right)^n \\
    &=\sum_{n=1}^{\infty} \frac{(-1)^{n-1}}{n} \left(\sum_{k=1}^{\infty}   \sum_{\mathbf{W}:|\mathbf{W}| =k} \prod_{\mathbf{V} \in P_{max}(\mathbf{W})} \left(\frac{\lambda_{\mathbf{V}}}{\mathbf{V}!}\mathcal{D}_\mathbf{V}\mathcal{E}_{\Gamma}^{Tr}[\tilde{\rho}^{(y)}]\right)\right)^n \label{eq:pinned_log_expand}
\end{align}
To estimate $\log(\mathcal{E}_{\Gamma}^{Tr}[\tilde{\rho}^{(y)}])$, we need to reorganize the above equation into a cluster expansion of $\log(\mathcal{E}_{\Gamma}^{Tr}[\tilde{\rho}^{(y)}])$, formally shown below.
\begin{align}\label{eq:pinned_log_expand_reorder}
    \log(\mathcal{E}_{\Gamma}^{Tr}[\tilde{\rho}^{(y)}]) &= \sum_{k=1}^{\infty}  \sum_{\mathbf{W} \in \mathcal{G}_k} \sum_{P \, \text{partitioning} \, \mathbf{W}}  C(P) \prod_{\mathbf{V} \in P} \left(\frac{\lambda_{\mathbf{V}}}{\mathbf{V}!}\mathcal{D}_\mathbf{V}\mathcal{E}_{\Gamma}^{Tr}[\tilde{\rho}^{(y)}]\right)
\end{align}
Note that all $\mathcal{D}_\mathbf{V}\mathcal{E}_{\Gamma}^{Tr}[\tilde{\rho}^{(y)}]$ commute because they are all diagonal, so we do not have to worry about the ordering in the multiplication. Following the same reasoning in the proof of Lemma \ref{lem:graph_coloring}, we identify $C(P)$ to the graph-coloring combinatorics.
\begin{equation}
    C(P) = \frac{1}{P!} \sum_{n=1}^{|P|} \frac{(-1)^n}{n}\chi^*(n,\text{Gra}(P))
\end{equation}
Plugging it back to Eq. (\ref{eq:pinned_log_expand_reorder})
\begin{align}
    \log(\mathcal{E}_{\Gamma}^{Tr}[\tilde{\rho}^{(y)}]) &= \sum_{k=1}^{\infty}  \sum_{\mathbf{W}: |\mathbf{W}|=k} \sum_{P  \in \text{PaC}(\mathbf{W})}  \left(\frac{1}{P!} \sum_{n=1}^{|P|} \frac{(-1)^n}{n}\chi^*(n,\text{Gra}(P)\right) \prod_{\mathbf{V} \in P} \left(\frac{\lambda_{\mathbf{V}}}{\mathbf{V}!}\mathcal{D}_\mathbf{V}\mathcal{E}_{\Gamma}^{Tr}[\tilde{\rho}^{(y)}]\right) \\
    &= \sum_{k=1}^{\infty}  \sum_{\mathbf{W}: |\mathbf{W}|=k} \frac{1}{\mathbf{W}!} \left( \sum_{B  \in \text{PaC}(\text{Gra}(\mathbf{W}))}   \sum_{n=1}^{|B|} \frac{(-1)^n}{n}\chi^*(n,\text{Gra}(B))\right) \prod_{\mathbf{V} \in P} \left(\lambda_{\mathbf{V}}\mathcal{D}_\mathbf{V}\mathcal{E}_{\Gamma}^{Tr}[\tilde{\rho}^{(y)}]\right)
\end{align}
Taking the cluster derivative selects the corresponding term
\begin{align}
    \mathcal{D}_{\mathbf{W}}\log(\mathcal{E}_{\Gamma}^{Tr}[\tilde{\rho}^{(y)}]) &= \left( \sum_{B  \in \text{PaC}(\text{Gra}(\mathbf{W}))}   \sum_{n=1}^{|B|} \frac{(-1)^n}{n}\chi^*(n,\text{Gra}(B))\right) \prod_{\mathbf{V} \in P} \left(\mathcal{D}_\mathbf{V}\mathcal{E}_{\Gamma}^{Tr}[\tilde{\rho}^{(y)}]\right)
\end{align}
Now we can upper-bound the operator norm
\begin{align}
    \norm{\mathcal{D}_{\mathbf{W}}\log(\mathcal{E}_{\Gamma}^{Tr}[\tilde{\rho}^{(y)}])} &\le \left( \sum_{B  \in \text{PaC}(\text{Gra}(\mathbf{W}))}   \sum_{n=1}^{|B|} \frac{(-1)^n}{n}\chi^*(n,\text{Gra}(B)) \right) \norm{\prod_{\mathbf{V} \in P} \left(\mathcal{D}_\mathbf{V}\mathcal{E}_{\Gamma}^{Tr}[\tilde{\rho}^{(y)}]\right)}\\
    &\le \left( \sum_{B  \in \text{PaC}(\text{Gra}(\mathbf{W}))}   \sum_{n=1}^{|B|} \frac{(-1)^n}{n}\chi^*(n,\text{Gra}(B)) \right) \prod_{\mathbf{V} \in P} \norm{\mathcal{D}_\mathbf{V}\mathcal{E}_{\Gamma}^{Tr}[\tilde{\rho}^{(y)}]}
\end{align}
Where in the second line we use the sub-multiplicity of the operator norm. It remains to upper-bound $\norm{\mathcal{D}_\mathbf{V}\mathcal{E}_{\Gamma}^{Tr}[\tilde{\rho}^{(y)}]}$. This is established in Lemma \ref{lem:pinned_deriv_bound}. Therefore, we have in the end
\begin{align}
    \norm{\mathcal{D}_{\mathbf{W}}\log(\mathcal{E}_{\Gamma}^{Tr}[\tilde{\rho}^{(y)}])} &\le \left( \sum_{B  \in \text{PaC}(\text{Gra}(\mathbf{W}))}   \sum_{n=1}^{|B|} \frac{(-1)^n}{n}\chi^*(n,\text{Gra}(B)) \right) \prod_{\mathbf{V} \in P} \beta^{|\mathbf{V}|} \\
    &\le \beta^{|\mathbf{W}|} \left( \sum_{B  \in \text{PaC}(\text{Gra}(\mathbf{W}))}   \sum_{n=1}^{|B|} \frac{(-1)^n}{n}\chi^*(n,\text{Gra}(B)) \right)
\end{align}
\end{proof}

\begin{proof}[Proof of Lemma \ref{lem:pinned_cluster_deriv_upper}]
First apply Lemma \ref{lem:pinned_graph_coloring}, then apply Lemma \ref{lem:combinatorial_estimate}.
\end{proof}

\subsection{Proof of Theorem \ref{thm:classical_decay_cmi}}
Finally, we put together the ingredients and complete the proof of Theorem \ref{thm:classical decay_cmi_informal}, formally Theorem \ref{thm:classical_decay_cmi}.
\begin{proof}[Proof of Theorem \ref{thm:classical_decay_cmi}]
    First, use the post-selection trick in Proposition \ref{prop:post_select} to convert CMI to an average of post-selected mutual information.
  \begin{equation}
    I(X:Z|\mathcal{T}[Y]) = \sum_y P(\mathcal{T}[Y]=y) I(X:Z|\mathcal{T}[Y]=y)
\end{equation}
To upper-bound CMI, it suffices to upper-bound each $I(X:Z|\mathcal{T}[Y]=y)$. Next, we apply Lemma \ref{lem:pinned_gibbs} to identify the post-selected marginal distribution on $XZ$ to the marginal of the pinned Gibbs distribution $\rho^{(y)} \propto e^{H^{(y)}}$. Then, we invoke Proposition \ref{prop:pinned_mi_cmi} and \ref{prop:pinned_mi_norm_bound} to have
\begin{equation}
    I(X:Z|\mathcal{T}[Y]=y) = I_{\mathcal{E}_{Y}^{Tr}[\rho^{(y)}]} (X:Z|Y) \le \norm{H^{(y)}(X:Z|\mathcal{E}_Y^{Tr}[Y])}
\end{equation}
In the remaining part, we will use the generalized cluster expansion to upper-bound $\norm{H^{(y)}(X:Z|\mathcal{E}_Y^{Tr}[Y])}$. We evaluate the cluster expansion of $\norm{H^{(y)}(X:Z|\mathcal{E}_Y^{Tr}[Y])}$.
    \begin{equation}
    \norm{H^{(y)}(X:Z|\mathcal{E}_Y^{Tr}[Y])} = \sum_{\mathbf{W}} \frac{\lambda_{\mathbf{W}}}{\mathbf{W}!}\mathcal{D}_\mathbf{W}\norm{H^{(y)}(X:Z|\mathcal{E}_Y^{Tr}[Y])}
\end{equation}
Using Lemma \ref{lem:pinned_conn_AC}, the non-trivial clusters are those connecting $XZ$ and are themselves connected. Denote the set of connected clusters that connect $XZ$ by $\mathcal{G}^{XZ}_m$. The minimal order of $m$ is $d_{XZ}$
\begin{equation}
     \norm{H^{(y)}(X:Z|\mathcal{E}_Y^{Tr}[Y])} = \sum_{m=d_{XZ}}^{\infty} \sum_{\mathbf{W} \in \mathcal{G}^{XZ}_m} \frac{\lambda_{\mathbf{W}}}{\mathbf{W}!}\mathcal{D}_\mathbf{W} \norm{H^{(y)}(X:Z|\mathcal{E}_Y^{Tr}[Y])}
\end{equation}
$H^{(y)}(X:Z|\mathcal{E}_Y^{Tr}[Y])$ contains four terms that are all of the form $\mathcal{E}_{\Gamma}^{Tr}[\tilde{\rho}^{(y)}]$, where $\mathcal{E}_{\Gamma}^{Tr}$ is the partial trace channel acting on $\Gamma$ and $\Gamma$ contains $Y$. Thus, we apply Lemma \ref{lem:pinned_cluster_deriv_upper} to upper-bound the norm of each term in the summation
\begin{align}
     \norm{H^{(y)}(X:Z|\mathcal{E}_Y^{Tr}[Y])} &\le 4 \sum_{m=d_{XZ}}^{\infty} \sum_{\mathbf{W} \in \mathcal{G}^{XZ}_m} (2e(\mathfrak{d}+1)\beta)^{m+1}
\end{align}
Following Lemma \ref{lem:num_conn_cluster}, the number of clusters in $\mathcal{G}^{XZ}_m$ is upper-bounded by $|\mathcal{G}^{XZ}_m| \le c \, \text{min}(|\partial X|, |\partial Z|) e\mathfrak{d}(1 + e(\mathfrak{d} - 1))^{m-1} $.\begin{align}
     \norm{H^{(y)}(X:Z|\mathcal{E}_Y^{Tr}[Y])} &\le c \, \text{min}(|\partial X|, |\partial Z|) \sum_{m=d_{XZ}}^{\infty}  (1 + e(\mathfrak{d} - 1))^{m-1} (2e(\mathfrak{d}+1)\beta)^{m+1}
\end{align}
Where we absorb the factor of four into $c$. The above series converges absolutely when $\beta \le \beta_c = 1/(2 e (\mathfrak{d}+1)(1+e(\mathfrak{d}-1)))$. Furthermore, the leading order term is $O((\frac{\beta}{\beta_c})^{d_{XZ}})$, so we have in the end,
\begin{align}
I(X:Z|\mathcal{T}[Y]) \le c \, \text{min}(|\partial X|, |\partial Z|) e^{\frac{d_{XZ}}{\xi}}
\end{align}
Where $\xi=O(\frac{\beta}{\beta_c})$ gives the Markov length.
\end{proof}

\section{Long-Range CMI at Low Temperature Through Fault-Tolerant MBQC}\label{low_temp_appn}
\subsection{Preliminaries}
Given a cluster state $\ket{\psi}_{ABC}$ encoding some fault-tolerant MBQC circuit with $k$ logical qubits. $B$ denotes the region to be measured in some local basis and $AC$ denotes the boundary region to be teleported (See Fig. \ref{fig:cluster}(a)). We apply the complete bit-flip channel $\mathcal{B} = \otimes_{i \in B} \mathcal{B}_i$ to $B$, and the resulting state $\mathcal{B}[\ketbra{\psi}{\psi}_{ABC}]$ has the following form.
\begin{equation}
   \mathcal{B}[\ketbra{\psi}{\psi}_{ABC}] = \sum_b p_b \ketbra{\psi_b}{\psi_b}_{AC} \otimes \ketbra{b}{b}_B
\end{equation}
Where we label the measurement outcome by $b$ and denote the post-measurement states on $AC$ to be $\ket{\psi_b}_{AC}$ and the corresponding probability to be $p_b$. The existance of $k$ logical qubits can be rephrased as follows: the Hilbert spaces $\mathcal{H}_A$ and $\mathcal{H}_C$ on $A$ and $C$ factorize into the code spaces and their complement
\begin{align}
    \mathcal{H}_A = \mathcal{H}_{A}^{(c)} \otimes \mathcal{H}_{A}^{(c\perp)}, \mathcal{H}_C = \mathcal{H}_{C}^{(c)} \otimes \mathcal{H}_{C}^{(c\perp)}
\end{align}
Where $\mathcal{H}_{A}^{(c)}$ and $\mathcal{H}_{C}^{(c)}$ are $2^k$ dimensional code spaces. to have $k$ logical qubits, $\ket{\psi_b}_{AC}$ has to be maximally entangled in the subspace $\mathcal{H}_{A}^{(c)} \otimes \mathcal{H}_{C}^{(c)}$. Specifically, let $\mathcal{E}^{Tr}_{AC,c\perp}$ be the channel that traces out $\mathcal{H}_{A}^{(c\perp)}$ and $\mathcal{H}_{C}^{(c\perp)}$, then $\mathcal{E}^{Tr}_{AC,c\perp}[\ketbra{\psi_b}{\psi_b}_{AC}]$ is maximally entangled on $\mathcal{H}_{A}^{(c)} \otimes \mathcal{H}_{C}^{(c)}$. This means that there exists a local unitary rotation $U_{b,C}$ acting on $C$ such that 
\begin{equation}
    (I_A \otimes U_{b,C})\mathcal{E}^{Tr}_{AC,c\perp}[\ketbra{\psi_b}{\psi_b}_{AC}](I_A \otimes U_{b,C}^\dag) = \ketbra{\Phi}{\Phi}_{AC}
\end{equation}
Where $\ket{\Phi} \frac{1}{2^k} \sum_x \ket{x}_{A^{(c)}}\ket{x}_{C^{(c)}} $ denotes the Bell state on $\mathcal{H}_{A}^{(c)} \otimes \mathcal{H}_{C}^{(c)}$. Equivalently, the post-selected mutual information on the code space saturates to $2k$ for all measurement outcome.
\begin{align}
    I_{\mathcal{E}^{Tr}_{AC,c\perp}[\ketbra{\psi_b}{\psi_b}_{AC}]}(A:C) = 2k, \forall b
\end{align}

Now we heat the Now cluster state to a finite temperature by applying local dephasing channels $\mathcal{N}_i$ acting on site $i$ with rate $p$, defined in Proposition \ref{prop:cluster_gibbs_noise}. These channel heats the cluster state to a finite temperature. After applying the noise channel $\mathcal{N} = \otimes_{i \in ABC} \, \mathcal{N}_i$ to $\ketbra{\psi}{\psi}_{ABC}$, the post-selected states become mixed.
\begin{equation}\label{eq:noisy_cluster_post_select}
   (\mathcal{N} \circ \mathcal{B}) [\ketbra{\psi}{\psi}_{ABC}] \sum_b p_b \rho_{AC|b} \otimes \ketbra{b}{b}_B
\end{equation}
Here we use $\rho_{AC|b}$ to denote the post-selected states. We now define the logical error rate below.
\begin{definition}
    \label{def:logical} Under the setup of Eq. (\ref{eq:noisy_cluster_post_select}), the logical error $q$ is defined as follows. Taking any $\rho_{AC|b}$, there exists a decoding channel $\mathcal{D}_{b} = \mathcal{D}_{b,A} \otimes \mathcal{D}_{b,C}$, where $\mathcal{D}_{b,A}$ and $\mathcal{D}_{b,C}$ act on $A$ and $C$, respectively. $\mathcal{D}_{b}$ depends on the measurement outcome $b$ in general. After applying the channel, $\mathcal{D}_b[\rho_{AC|b}]$ has fidelity $1-q$ to the Bell state $\ket{\Phi}$ on the code space on average. In other words,
\begin{equation}\label{eq:average_fidelity_bound}
    \sum_b p_b F((\mathcal{E}^{Tr}_{AC,c\perp} \circ \mathcal{D}_b)[\rho_{AC|b}],\ketbra{\Phi}{\Phi}) = 1-q
\end{equation}
\end{definition}
We note that the above definition lower-bounds the logical error rate extracted in the following type of Clifford simulations \cite{gidney2021stim}: simulating each $\mathcal{N}_i$ by applying a $Z$ gate on site $i$ with probability $p$, and average over multiple trials. In this way, for each trial we always keep track of a pure Clifford state which is computationally tractable. We label the patterns of $Z$ gates applied with $m$ and denote the state on $AC$ under the measurement outcome $b$ and pattern $m$ as $\ket{\psi_{b,m}}_{AC}$. $\rho_{AC|b}$ is the same as the average of all $\ketbra{\psi_{b,m}}{\psi_{b,m}}_{AC}$ over $m$.
\begin{equation}
    \rho_{AC|b} = \sum_{m,b} p_m p_{b|m} \ketbra{\psi_{b,m}}{\psi_{b,m}}_{AC}
\end{equation}
Where $p_m$ denotes the probability that pattern $m$ occurs in the simulation and $p_{b|m}$ denotes the probability of measurement outcome $b$ conditioned on the pattern $m$. After obtaining $\ket{\psi_{b,m}}_{AC}$ in each trial, the simulation applies $\mathcal{D}_{b}$ and evaluates the fidelity to the Bell state in the logical subspace, that is, $F((\mathcal{E}^{Tr}_{AC,c\perp} \circ \mathcal{D}_b)[\ketbra{\psi_{b,m}}{\psi_{b,m}}_{AC}],\ketbra{\Phi}{\Phi})$. The logical error rate in such Clifford simulations, denoted as $q_c$, is then defined as the average of individual fidelity over all patterns and measurement outcomes.
\begin{equation}
    1-q_c = \sum_{b,m} p_m p_{b|m} F((\mathcal{E}^{Tr}_{AC,c\perp} \circ \mathcal{D}_b)[\ketbra{\psi_{b,m}}{\psi_{b,m}}_{AC}],\ketbra{\Phi}{\Phi})
\end{equation}
We now show that $q_c$ is always greater than $q$, so that having an exponentially small $q_c$ implies an exponentially small $q$
\begin{lemma}
    The logical error rate $q_c$ from Clifford simulations upper-bounds the logical error rate $q$.
    \begin{equation}
    q \le q_c
\end{equation}
\end{lemma}
\begin{proof}
    The above Lemma is equivalent to show the following relation in fidelity.
    \begin{equation}\label{eq:fidelity_logical_clifford}
   1-q = F((\mathcal{E}^{Tr}_{AC,c\perp} \circ \mathcal{D}_b)[\rho_{AC|b}],\ketbra{\Phi}{\Phi}) \ge \sum_m p_m p_{b|m} F((\mathcal{E}^{Tr}_{AC,c\perp} \circ \mathcal{D}_b)[\ketbra{\psi_{b,m}}{\psi_{b,m}}_{AC}],\ketbra{\Phi}{\Phi}) = 1-q_c
\end{equation}
This can shown easily shown via the joint concavity of fidelity: given $\rho_1$, $\rho_2$, $\sigma_1$, and $\sigma_2$ and a $\lambda \in [0,1]$,
\begin{equation}
    F(\lambda \rho_1 + (1-\lambda) \rho_2, \lambda \sigma_1 + (1-\lambda) \sigma_2) \ge \lambda F(\rho_1,\sigma_1) + (1-\lambda) F(\rho_2,\sigma_2)
\end{equation}
We set $\rho_1=\rho_2=\ketbra{\Phi}{\Phi}$ and set different $\sigma_i$ to $\ketbra{\psi_{b,m}}{\psi_{b,m}}_{AC}$. By applying the joint concavity iteratively, we obtain Eq. (\ref{eq:fidelity_logical_clifford}).
\end{proof}
In some sense, the logical error rate defined in Definition \ref{def:logical} is the ``intrinsic'' logical error rate associated with the noise model and the decoder. On the other hand, the Clifford simulations decompose the noisy dynamics into Clifford trajectories and only provides an upper-bound on the intrinsic logical error rate.

\subsection{Proof of Theorem \ref{thm:long_range_cmi_informal} (formally Theorem \ref{thm:long_range_cmi})}
We now prove Theorem \ref{thm:long_range_cmi_informal}, formalized below.
\begin{theorem}\label{thm:long_range_cmi}
    Given a cluster state $\ket{\psi}_{ABC}$ that encodes some fault-tolerant MBQC circuit with $k$ logical qubits. $B$ denotes the region to be measured. Under the noisy measurement channel with measurement error rate $p$, suppose that $\ket{\psi}_{ABC}$ has logical error rate $q$.
    
Consider the corresponding finite-temperature cluster state $\rho^{(\beta)}_{ABC}$ with $\beta = \frac{1}{2} \log(p^{-1}-1)$ and apply product of local bit-flip channel acting on $B$ $\mathcal{B} = \mathcal{B} _1 \otimes \mathcal{B} _2 \otimes \ldots \otimes \mathcal{B} _{|B|}$. After applying the channel, the quantum hidden Markov network $\mathcal{B}[\rho^{(\beta)}_{ABC}]$ satisfies the following lower-bound on its long-range CMI:
    \begin{align}
        I_{\mathcal{B}[\rho^{(\beta)}_{ABC}]}(A:C|B) \ge 2k - O(k \sqrt{q} + q^{\frac{1}{6}})
    \end{align}
\end{theorem}
When $q=e^{(-\Omega(n))}$, the last two terms are exponentially small in $n$, so we arrive at Theorem \ref{thm:long_range_cmi_informal} stated in the main text. Note that the power of $\frac{1}{6}$ is an artifact originating from the proof technique and has no physical significance. In fact, we can choose some other exponents, but when $q=e^{(-\Omega(n))}$, these small numbers always decay exponentially.
\begin{proof}
We first use Proposition \ref{prop:cluster_gibbs_noise} to establish the equivalence between the finite-temperature cluster state $\rho^{(\beta)}_{ABC}$ and the noisy cluster state $\mathcal{N}[\ket{\psi}\bra{\psi}_{ABC}]$. After taking the decomposition in Eq. (\ref{eq:noisy_cluster_post_select}), realize that CMI becomes a classical average of post-measurement mutual information because $B$ is now classical
\begin{equation}
        I_{\mathcal{B}[\rho^{(\beta)}_{ABC}]}(A:C|B) = \sum_b p_b I_{\rho_{AC|b}}(A:C)
\end{equation}
We will lower-bound the averaged $I_{\rho_{AC|b}}(A:C)$ using the fidelity bound in Eq. (\ref{eq:average_fidelity_bound}). Using the relation between trace distance $D(\rho, \sigma)$ and fidelity $F(\rho, \sigma)$: $D(\rho, \sigma) \le \sqrt{1-F(\rho, \sigma)}$, we have
\begin{equation}\label{eq:average_distance_bound}
    \sum_b p_b D((\mathcal{E}^{Tr}_{AC,c\perp} \circ \mathcal{D}_b)[\rho_{AC|b}],\ketbra{\Phi}{\Phi}) \le \sqrt{q}
\end{equation}
Next, we invoke the Fannes–Audenaert inequality which relates the trace distance to the entropy difference:
\begin{equation}
    |S(\rho) - S(\sigma)| \le D(\rho, \sigma) \log(d) + H(D(\rho, \sigma))
\end{equation}
Where $d$ denotes the Hilbert space dimension and $H(D(\rho, \sigma))$ is the Shannon entropy associated with the trace distance.
\begin{equation}
  H(D(\rho, \sigma)) =  -D(\rho, \sigma) \log_2(D(\rho, \sigma)) - (1-D(\rho, \sigma)) \log_2((1-D(\rho, \sigma)))
\end{equation}
Applying the Fannes–Audenaert inequality to the mutual information, we have
\begin{equation}\label{eq:cmi_fa_bound}
    \sum_b p_b I_{(\mathcal{E}^{Tr}_{AC,c\perp} \circ \mathcal{D}_b)[\rho_{AC|b}]}(A^{(c)}:C^{(c)}) \ge 2k - \sqrt{q} 2 |A^{(c)}C^{(c)}| - \sum_b p_b H(D_b) =  2k - 4k \sqrt{q} - \sum_b p_b H(D_b)
\end{equation}
Where $|A^{(c)}C^{(c)}|$ denotes the logarithm of the Hilbert space dimension of $A^{(c)}C^{(c)}$ which is $2k$. We use the short-handed notation $D_b$ to denote $D((\mathcal{E}^{Tr}_{AC,c\perp} \circ \mathcal{D}_b)[\rho_{AC|b}],\ketbra{\Phi}{\Phi})$ and use $H(D_b)$ to denote the Shannon entropy associated with $D_b$. To upper-bound the averaged $H(D_b)$, first notice that $D_b$ is always non-negative, so we apply the Markov's inequality
\begin{equation}
    P(D_b \ge d') \le \frac{\sqrt{q}}{d'}
\end{equation}
Where $d'$ is a small parameter we will choose to make the entropy bound as tight as possible (for this problem we will focus on the regime where $d' \le 0.5$). Therefore, we can upper-bound the averaged $H(D_b)$ by the following expression
\begin{equation}
    \sum_{b} p_b H(D_b) \le \frac{\sqrt{q}}{d'} + (1-\frac{\sqrt{q}}{d'}) H(d')
\end{equation}
Here, the first term corresponds to setting all $D_b \ge d'$ to be $0.5$ so $H(D_b)=1$, and setting all $D_b \le d'$ to be exactly $d'$ (remember $H(d')$ grows monotonically when $d' \le 0.5$). This maximizes the Shannon entropy in both cases.

When $d' \le 0.5$, we have upper-bound $H(d')$ by $-2 d' \log_2(d')$
\begin{equation}
    \sum_{b} p_b H(D_b) \le \frac{\sqrt{q}}{d'} - 2 (1-\frac{\sqrt{q}}{d'}) d' \log_2(d') \le \frac{\sqrt{q}}{d'} - 2 d' \log_2(d') 
\end{equation}
$-2 d' \log_2(d')$ grows slower then any power law $x^\alpha$, where $\alpha \le 1$, when $d'$ is small, so we choose an arbitrary power-law bound: $-2 d' \log_2(d') \le 3 \sqrt{d'}$.
\begin{equation}
    \sum_{b} p_b H(D_b) \le \frac{\sqrt{q}}{d'} + 3 \sqrt{d'}
\end{equation}
Now we can choose a $d'$ to minimize the above upper bound. This happens at $d'= (\frac{2}{3})^{2/3} q^{1/3}$. Correspondingly,
\begin{equation}
    \sum_{b} p_b H(D_b) \le 3 (\frac{3}{2})^{\frac{2}{3}}q^{\frac{1}{6}}= O(q^{\frac{1}{6}})
\end{equation}
Plugging this back to the CMI upper bound (Eq. (\ref{eq:cmi_fa_bound})),
\begin{equation}
    \sum_b p_b I_{(\mathcal{E}^{Tr}_{AC,c\perp} \circ \mathcal{D}_b)[\rho_{AC|b}]}(A^{(c)}:C^{(c)}) \ge  2k - O(k \sqrt{q} + q^{\frac{1}{6}})
\end{equation}
Finally, given that $\mathcal{E}^{Tr}_{AC,c\perp} \circ \mathcal{D}_b$ is a product of local channels acting on $A$ and $C$ separately, by data processing inequality $I_{\mathcal{B}[\rho^{(\beta)}_{ABC}]}(A:C|B)$ can only be bigger.
\begin{equation}
    I_{\mathcal{B}[\rho^{(\beta)}_{ABC}]}(A:C|B) \ge \sum_b p_b I_{(\mathcal{E}^{Tr}_{AC,c\perp} \circ \mathcal{D}_b)[\rho_{AC|b}]}(A^{(c)}:C^{(c)}) \ge 2k - O(k \sqrt{q} + q^{\frac{1}{6}})
\end{equation}
\end{proof}Lastly, we point out that Ref \cite{raussendorf2005long} has already argued for the existence of long-range localizable entanglement in noisy three-dimensional cluster state. Our result essentially generalizes their arguments to arbitrary fault-tolerant MBQC protocols encoded in finite-temperature cluster states. Also, we note that the choice of the cluster state Hamiltonian (Eq. \ref{eq:cluster_hamiltonian}) is essential here to connect the finite-temperature cluster states to the local noise model in MBQC. In general, one can choose an arbitrary set of local stabilizer generator of the cluster state as its parent Hamiltonian, but a generic choice does not have the property, and the finite-temperature cluster states can correspond to the non-local noise models in MBQC.

\subsection{Comparison with Ref. \cite{negari2024spacetime}}\label{comp_negari}
In this subsection, we compare our above result with Ref. \cite{negari2024spacetime}. The authors there also consider a MBQC circuit that foliates a (classical) error-correcting code. They show that the error-correcting threshold, after translated into a temperature using Proposition \ref{prop:cluster_gibbs_noise}, is identical to the critical point where the Markov length diverges. However, we note that while both works construct the same model and probe the same information-theoretic transition, the tripartition used for evaluating CMI is different. In their setup, $\textcolor{Cerulean}{A}\textcolor{Green}{B}\textcolor{Goldenrod}{C}$ are arranged in an annulus geometry (Fig. \ref{fig:setup}(a,d)). This setup probes the \emph{local} CMI, where the Markov length is finite at both high and low temperature and only diverges at the critical point. Meanwhile, in our setup, $\textcolor{Cerulean}{A}\textcolor{Green}{B}\textcolor{Goldenrod}{C}$ globally partitions the system into the bulk and boundaries and probes the \emph{global} CMI. In the low-temperature phase, this CMI is large, reflecting the teleportation of logical information; in the high-temperature phase, this CMI decays with a finite Markov length.

The distinctive behaviors between the two partitions can be understood from their operational meanings. In \cite{negari2024spacetime}, they show that finite Markov length under the annulus geometry implies the existence of a quasi-local decoder for the MQBC, and the divergence of the Markov length reflects the failure in decoding. This is also why the temperature where Markov length diverges has to coincide with the error threshold.

On the other hand, CMI under our global partition directly reflects the teleportation of logical information, and the existence of long-range CMI at low temperatures reflects the survival of logical information. Therefore, the two CMI quantities under different partitions have different operational meaning. While both quantity signifies the same information-theoretic phase transitions, they reflect fundamentally different properties.

\section{Mixed-State Phases at High Temperatures}\label{recovery}
In this section, we discuss the implication of finite Markov length in the context of mixed-state phases of matter. Specifically, we will show that all previously described high-temperature Markov networks and hidden Markov networks are in the trivial phase.

We will invoke the notion of Lindbladian in defining mixed-state phases of matter. A Lindbladian describes the dynamics of a mixed quantum state in contact with a Markovian environment.
\begin{equation}
    \frac{d}{d t} \rho = \mathcal{L}[\rho]
\end{equation}
Where the Lindbladian $\mathcal{L}$ cane be written as
\begin{equation}
    \mathcal{L}[\rho] = -i [B, \rho] + \sum_a  A^{a} \rho  A^{a \dag} - \frac{1}{2} \{A^{a\dag} A^{a}, \rho \}
\end{equation}
Here, $B$ denotes a Hermitian matrix and $A^{a}$ are called jump operators. 

Following Ref. \cite{coser2019classification}, we define two states $\rho_1$ and $\rho_2$ to be in the same phase if there exists two Lindbladians $\mathcal{L}_{12}$ and $\mathcal{L}_{21}$, in both of which all $A^a$ have $O(\text{polylog}(n))$ support and $B$ is a sum of terms with $O(\text{polylog}(n))$ support, such that
\begin{align*}
    \norm{ e^{ t \mathcal{L}_{12}}[\rho_1] - \rho_2 }_1 &= \epsilon_{12} \\
    \norm{ e^{ t \mathcal{L}_{21}}[\rho_2] - \rho_1 }_1 &= \epsilon_{21}
\end{align*}
Where $t=O(\text{polylog}(n))$ and $\epsilon_{12},\epsilon_{21}$ scale to zero as $n$ scales to infinity. $\norm{ \cdot}_1$ denotes the Schatten-one norm. In other words, two states are in the same phase if they are connected by Lindbladians with quasi-local operators and evolving for poly-logarithmic time at most. 

Another way is to replace the Lindbladian with a circuit of quasi-local quantum channels, shown in Fig. \ref{fig:phase}(a). Each grey box represents a quantum channel with $O(\text{polylog}(n))$ support and the circuit depth is $O(\text{polylog}(n))$. The previous Lindbladian construction can be converted to this channel construction via trotterization. This is essentially the mixed-state equivalent of finite-depth unitary circuits in classifying phases of matter in ground states. The connection has to be two-way, as any state is connected to the maximally mixed state through applying single-qubit depolarization channels on all sites.

Theorem \ref{thm:decay_cmi_informal} and \ref{thm:classical decay_cmi_informal} directly implies the existence of quasi-local quantum channels connecting high-temperature hidden Markov networks. 

\begin{corollary}
Consider a high-temperature commuting Gibbs state $\rho_\beta$. Suppose we have a channel $\mathcal{E}_p = \otimes_i \mathcal{E}_{i,p}$ parametrized by some parameter $p \in [0,1]$. For example, $\mathcal{E}_{i,p}$ could be the dephasing or depolarizing channel on site $i$ with rate $p$. We demand that $\mathcal{E}_p$ to be unital and commutation-preserving for all $p$ if $\rho_\beta$ is not diagonal in the computational basis. 

Under the above condition, Theorem \ref{thm:decay_cmi_informal} and \ref{thm:classical decay_cmi_informal} state that $\mathcal{E}_p[\rho_\beta]$ has a finite Markov length for all $p$ and $\beta \le \beta_c$. Applying the main conclusion in Ref. \cite{sang2024stability}, there exists a quasi-local quantum channel $\mathcal{R}_E$ where each channel has $O(\log(n))$ support and the circuit has $O(\log(n))$ depth. $\mathcal{R}$ brings $\mathcal{E}_p[\rho_\beta]$ back to $\rho_\beta$ up to inverse polynomial precision.
\begin{equation}
   \norm{(\mathcal{R}_E \circ \mathcal{E}_p)[\rho_\beta] - \rho_\beta}_1 = O(\frac{1}{\text{poly}(n)}), \forall p, \beta \le \beta_c
\end{equation}
\end{corollary}
The channel $\mathcal{R}_E$ is constructed from the recovery map \cite{fawzi2015quantum,wilde2015recoverability,sutter2016universal,junge2018universal}. The condition that $\mathcal{E}_p$ being unital and commutation-preserving for all $p$ is true for depolarization channel $\mathcal{D}_p = \otimes_i \mathcal{D}_{i,p}$, defined as
\begin{equation}
    \mathcal{D}_{i,p}[\rho] = (1-p) \rho + p \frac{I}{q}
\end{equation}
Given that $\mathcal{D}_1 [\rho_\beta] \propto I$ gives the infinite-temperature Gibbs state, the above corollary implies a quasi-local channel $\mathcal{R}_D$ that brings the infinite-temperature Gibbs state back to $\rho_\beta$. Thus, we show that all high-temperature commuting Gibbs states are in the same phase as the infinite-temperature Gibbs state. In addition, the above argument also shows that high-temperature hidden Markov network, under the channel restrictions, are in the same phase are the infinite-temperature Gibbs state. On the other hand, we know there are low-temperature Gibbs states, e.g. low-temperature Ising models that possess long-range correlations, are in different phases, so the Markov length has to diverge when crossing the phase boundaries. We depict the phase diagram and the quasi-local quantum channels interpolating different states in Fig. \ref{fig:phase}(b).

However, another more common way to preprate Gibbs states starting from the maximally-mixed state is via Gibbs sampling, so we discuss under what circumstances do the Gibbs samplers satisfy the quasi-locality condition that defines a phase. 

As a quickly review, we aim to construct a Lindbladian satisfying detailed balance and ergodicity in Gibbs sampling. Detailed balance ensures that the Gibbs state $\rho_\beta \propto e^{-\beta H}$ is a steady state of $\mathcal{L}$, that is $\mathcal{L}[\rho_\beta]=0$. Ergodicity ensures that the steady state of $\mathcal{L}$ is unique. Detailed balance and ergodicity together establishes that $\mathcal{L}$ generates $\rho_\beta$ in the steady state.

We will apply $\mathcal{L}$ to the maximally mixed state and discuss whether the space and time cost satisfies the requirement of mixed-state phases. Remarkably, Ref. \cite{chen2023efficient,chen2023quantum} show that operators in $\mathcal{L}$ can be chosen to have $O(\log(n))$ support while still satisfying the detailed balance condition. Ergodicity is guaranteed by choosing $A^{a}$ to generate the complete set of operators. Therefore,  it remains to upper-bound the time it takes for the Lindbladian to approach the steady state, often called the mixing time. 

To satisfy the the definition of mixed-state phases, we would need the mixing time to be $O(\text{polylog}(n))$, a property called rapid-mixing. Rapid mixing has been established for classical Gibbs states at high-temperature, in the regime where correlation decays exponentially \cite{martinelli1999lectures,guionnet2003lectures}. This is also the regime where cluster expansion converges, so our result does not improve over the previous result. In the case of commuting Hamiltonian, rapid mixing has only been established for 2-local Hamiltonian \cite{kochanowski2024rapid}, whereas our result is not sensitive to the degree of the Hamiltonian. Recently, Ref. \cite{rouze2024optimal} establishes the rapid mixing of non-commuting Gibbs states at high temperature. However, while their temperature bound is $O(1)$, it is significantly higher than our temperature bound. To sum up, our result neither contradicts nor is contained in previous works but is rather complementary.

Lastly, we point out that high-temperature Gibbs state can be written as a mixed product state \cite{bakshi2024high}. Moreover, the authors there gives a classical polynomial-time algorithm to sample the product state, thus giving an efficient algorithm to prepare high-temperature Gibbs states. Ref. \cite{yin2023polynomial} also gives a classical algorithm to sample from the high-temperature Gibbs state directly. Both results apply in the temperature range where cluster expansion converges, rendering the quasi-local recovery map unfavorable in practice. Nevertheless, the above two algorithms are spatially non-local, and therefore irrelevant to the analysis of mixed-state phases. 

\begin{figure}
\includegraphics[width=0.8\linewidth]{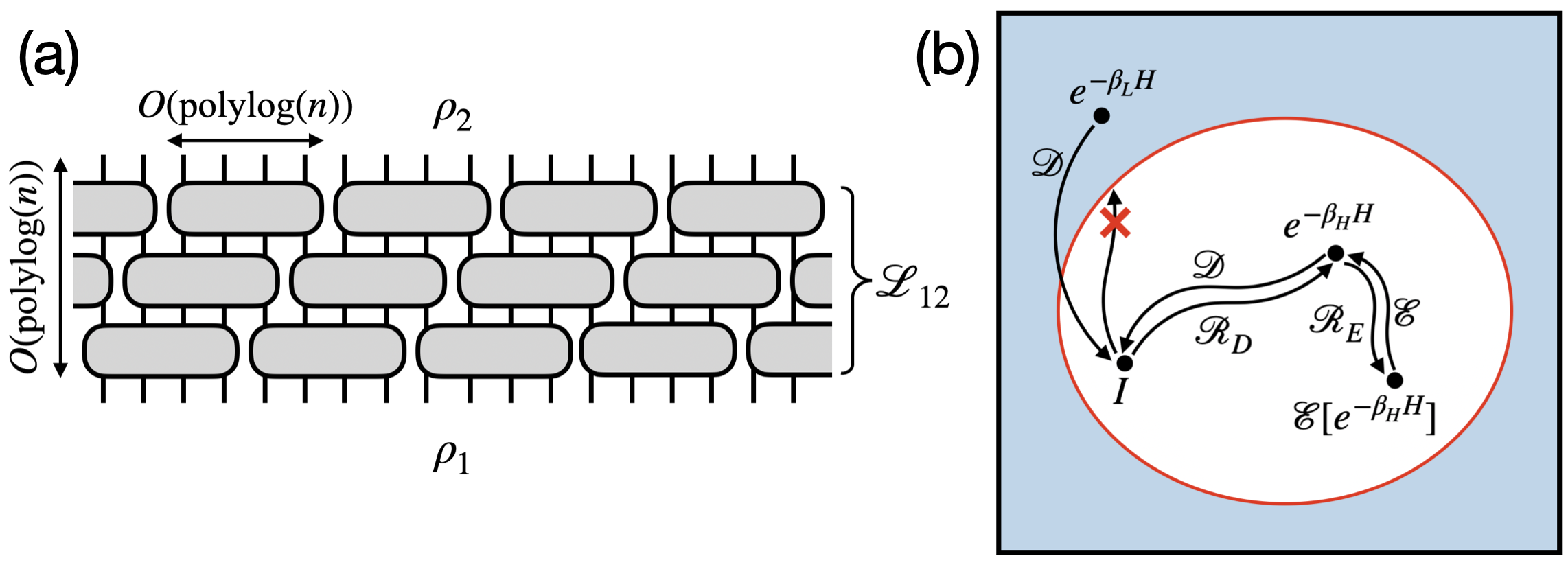}
\caption{\label{fig:phase}(a) A Quasi-local channel that connects two states. (b) A schematics of mixed-state phase diagram}
\end{figure}

\section{Finite Markov Length Implies Efficient Neural Quantum State Representations}\label{neural}
In this section, we show that high-temperature hidden Markov networks admit efficient neural quantum state representations. In Ref. \cite{yang2024can}, the author shows that for any state $\rho$, if the measurement outcome distribution has a finite Markov length (they call it CMI length), then it can be encoded into a quasi-polynomial-sized feedforward and recurrent neural network. Since measurements in the computational basis are equivalent to applying dephasing channels, our results implies efficient neural quantum state representations when $\rho$ high-temperature hidden Markov networks, under the commutation-preserving condition.
\begin{corollary}
Given a high-temperature commuting Gibbs state $\rho_\beta$ and local channels $\mathcal{E} = \otimes_i \mathcal{E}_i $. Denote $\mathcal{D} = \otimes_i \mathcal{D}_i $ as the channel that completely dephases all qubits. Under the condition of Theorem \ref{thm:decay_cmi} and \ref{thm:classical_decay_cmi}, suppose $\{\mathcal{D}_i \circ \mathcal{E}_i\}$ is also commutation-preserving, then the measurement outcome distribution $(\mathcal{D} \circ \mathcal{E})[\rho_\beta]$ can be represented as a quasi-polynomial-sized feedforward and recurrent neural network, following Theorem B.3 in Ref. \cite{yang2024can}.  
\end{corollary}
We note that while the result in Ref. \cite{yang2024can} is stated for pure states, the result can be generalized to mixed states as well since the construction only depends on the finite Markov length of the measurement outcome distribution, not on the purity of the underlying quantum states.
\end{widetext}

% \nocite{*}

\bibliography{apssamp}% Produces the bibliography via BibTeX.

\end{document}